\documentclass[prodmode,acmtompecs]{acmsmall} 

\usepackage{graphicx}
\usepackage[cmex10]{amsmath}	
\usepackage{amssymb}	
\usepackage{multirow}

\DeclareMathOperator*{\argmax}{arg\,max} %
\DeclareMathOperator*{\expectation}{\mathrm{E}} %

\newcommand{\reqvec}{\mathbf{q}}
\newcommand{\reqscalar}{q}
\newcommand{\ribvec}{\boldsymbol{\alpha}}
\newcommand{\ribscalar}{\alpha}

\newcommand{\priorityImplication}{monotonicity in individual expected payoff}
\newcommand{\modelSatisfyPriorityImplication}{a system satisfies monotonicity in payoffs}

\newcommand{\concaveHull}{R}
\newcommand{\feasibilityRegion}{F}
\newcommand{\fullUserSet}{N}

\newcommand{\succS}[1]{\succ_{#1}}

\newcommand{\precS}[1]{\prec_{#1}}

\newcommand{\myComments}[1]{}	

\newif\ifinfocom
\infocomfalse

\newif\ifextended
\extendedtrue

\newif\ifdissertation
\dissertationfalse

\newif\ifhuawei
\huaweifalse

\newcommand{\infocomStart}{\ifinfocom \myComments{Infocom: }}
\newcommand{\extendedStart}{\ifextended  \myComments{Extended version: }}
\newcommand{\dissertationStart}{\ifdissertation  \myComments{Dissertation version: }}
\newcommand{\huaweiStart}{\ifhuawei  \myComments{Huawei version: }}

\newcommand{\commentEnd}{\myComments{End}}

\newcommand{\add}[1]{#1}

\begin{document}

\markboth{Y. Du and G. de Veciana}{Efficiency and Optimality of Largest Deficit First Prioritization}

\title{Efficiency and Optimality of Largest Deficit First Prioritization: Resource Allocation for Real-Time Applications}
\author{YUHUAN DU and GUSTAVO DE VECIANA
\affil{The  University of Texas at Austin}
}

\begin{abstract}
An increasing number of real-time applications with compute and/or communication deadlines
are being supported on shared infrastructure.  Such applications can often tolerate 
occasional deadline violations without substantially impacting their Quality of Service (QoS).
A fundamental problem in such systems is deciding how
to allocate shared resources so as to meet applications' QoS requirements.
A simple framework to address this problem is to, (1) dynamically prioritize users as a possibly complex function of their deficits (difference of achieved vs required QoS), and (2) allocate resources so to expedite users with higher priority.
This paper focuses on a general class of systems using such priority-based resource allocation.
We first characterize the set of feasible QoS requirements and show the optimality of max weight-like prioritization.
We then consider simple weighted Largest Deficit First ($\mathbf{w}$-LDF)
prioritization policies, where users with higher weighted QoS deficits are given higher priority.
The paper gives an inner bound for the feasible set under $\mathbf{w}$-LDF policies, and,
under an additional monotonicity assumption, characterizes its geometry leading to a sufficient condition for optimality.
Additional insights on the efficiency ratio of $\mathbf{w}$-LDF policies, the 
optimality of hierarchical-LDF and characterization of clustering of failures are also discussed.
\end{abstract}


\begin{bottomstuff}
\add{
This research was supported by Huawei Technologies Co. Ltd.}

\add{A conference version of this paper has been accepted to INFOCOM 2016. 
}
\end{bottomstuff}

\maketitle

\section{Introduction}
\label{sec_introduction}
A growing number of real-time applications with compute and/or communication deadlines
are being moved onto shared infrastructure, e.g., ranging from embedded systems to efficient cloud infrastructure.
Such applications include control, multimedia processing, and/or machine learning components associated
with enabling various types of user services as well as wireless, intelligent transportation and energy systems.
In many cases such applications can tolerate occasional deadline violations, i.e., have soft 
constraints,  without impacting the application Quality of Service (QoS). For example, 
applications with feedback can quickly compensate for errors, or humans may tolerate occasional 
failures in video processing since they can be partially concealed, or wireless base stations can
tolerate occasional frame losses, since these can be retransmitted.
More generally real-time applications' long-term QoS may depend in a complex manner 
on what was accomplished on time, e.g., partial completion of a set of tasks, or notions of video quality. 

Enabling efficient sharing of compute/communication resources is a challenging problem. 
On the one hand, even for a single resource, tying the sharing model, e.g., 
round robin, priority schemes, to QoS metrics is generally hard due
to the uncertainty in applications' workloads and possible variations in processing speeds.
On the other hand, today's applications leverage complex networks of heterogeneous compute/communication resources, 
e.g., multi-core computers, embedded network system, or combinations of computation on mobile devices 
and the cloud. 
Consider the example in Figure~{\ref{fig_complex_network_resources}}. 
User 1 periodically generates a task that needs to be processed sequentially on Resources A, B, C, D in each 
period while User 2 generates tasks to be processed on Resource B then C. 
How should one go about designing resource sharing policies 
across multiple heterogeneous resources, where parallelism, task preemption and migration are allowed?
Furthermore, how can one address heterogeneous QoS requirements associated  
with real-time applications? For example, User 1's QoS may still benefit from partial completions while User 2 
only benefits if all processing is completed. 
This general class of problems involving both heterogeneous resources and user QoS requirements
is the focus of this paper. 

\infocomStart
\begin{figure}[htp]
  \centering
  \includegraphics[width=0.23\textwidth]{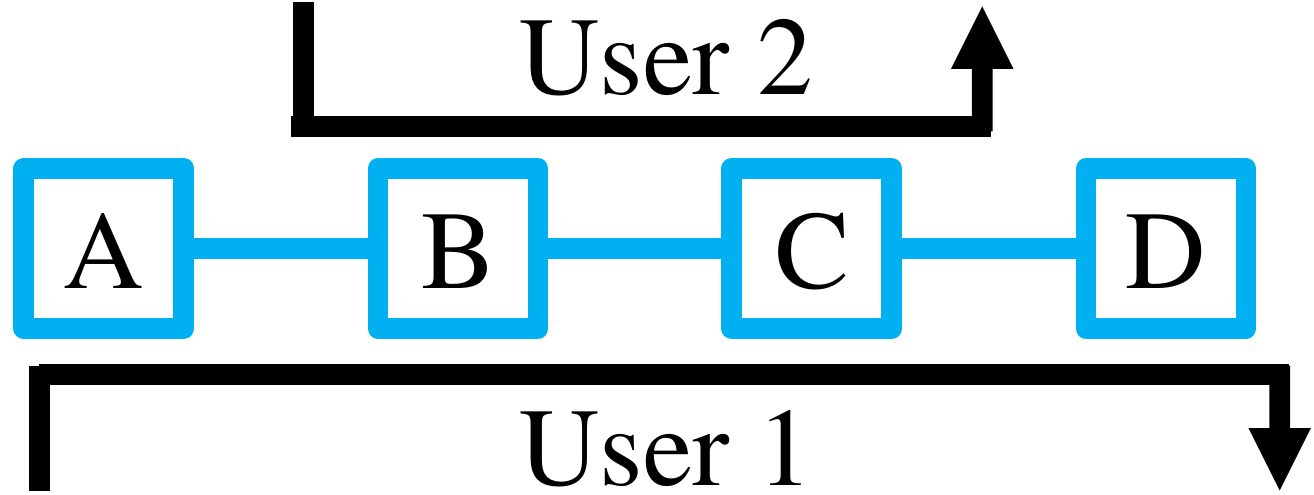}
  \caption{An example for a network of resources. A, B, C and D represent compute/communication resources. }
  \label{fig_complex_network_resources}
\end{figure}
\commentEnd\fi

\extendedStart
\begin{figure}[htp]
  \centering
  \includegraphics[width=0.32\textwidth]{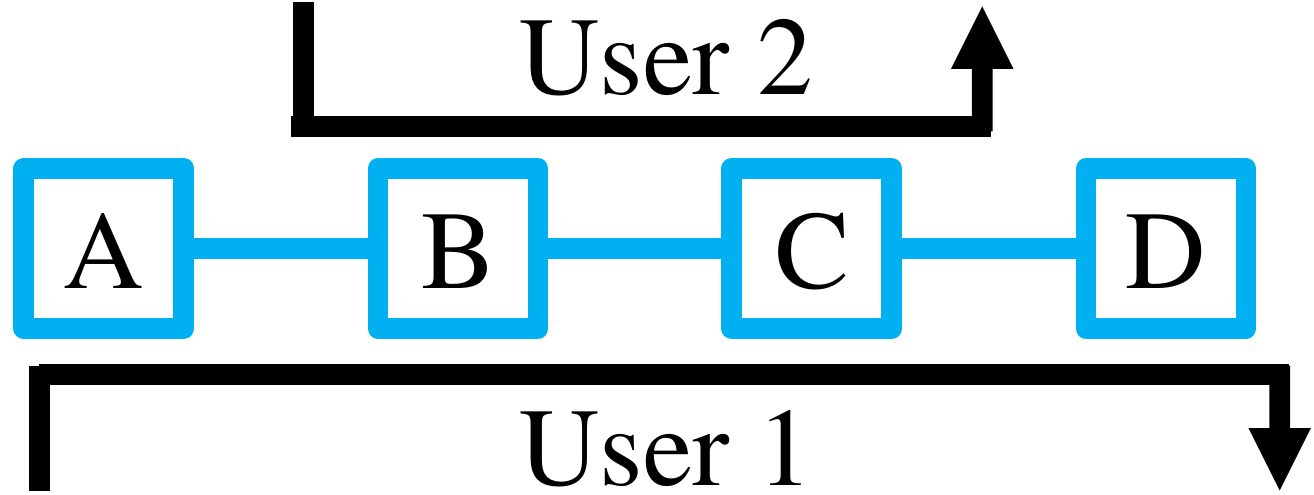}
  \caption{An example for a network of resources. A, B, C and D represent compute/communication resources. Tasks from User 1 need to be processed on A, B, C, D while tasks from User 2 require processing on B, C. }
  \label{fig_complex_network_resources}
\end{figure}
\commentEnd\fi


The design space of possible solutions to this problem is huge and has been
explored in many research communities. In this paper we study an approach to
resource allocation based on a decomposition of concerns:
\begin{enumerate}
\item user priorities are dynamically set based on the history outcomes; 
\item and, resources are allocated so as to favor users with higher priority.
\end{enumerate}
In such a framework there is quite a bit of latitude in choosing how priorities are set, and
in turn how these affect the allocation of resources. 
For example, users' priorities could be set based on measured deficits, the ``difference'' of the required and achieved QoS, 
i.e., Largest Deficit First (LDF) prioritization. 
In turn, for complex systems such as that in Figure~{\ref{fig_complex_network_resources}}, resources
could be allocated greedily giving preemptive access to tasks associated with higher-priority users.

In general an optimal user prioritization strategy could leverage detailed
information regarding how these priorities will impact the allocation of resources and completion outcomes to achieve
the best possible user QoS. 
Such strategies require excessive amounts of information regarding the underlying compute/communication resources and resource allocation mechanism, and thus are generally hard to implement. 
By contrast, LDF-based prioritization is quite intuitive. 
It requires only tracking of users' possibly heterogeneous QoS deficits, in this
sense it is truly decoupling user prioritization from the underlying priority-based resource allocation. 
Unfortunately, it is known to be suboptimal in certain settings \cite{DiW06,JLS07,KWJ13}.

A theoretical study of the efficiency and, possibly optimality, of LDF-based prioritization 
systems supporting real-time users with heterogeneous QoS requirements 
is the main focus of this paper. We note, however, that we do not directly address the design
of the underlying priority-based resource allocation, although we consider some natural
characteristics it could have to ensure optimality when combined with LDF user prioritization.

{\bf \em Related Work. }
There have been much work studying dynamic prioritization policies in the context of diverse resource, workload and/or QoS models. 

The authors in \cite{HoK12,HoK13b} propose a framework to model a wireless access point serving a set of clients that in each period generate packets which need to be transmitted by the end of the period. 
In their model only one client can transmit at a time and thus the access point can be viewed as a single resource. Each client transmits its packets over an unreliable channel which has a fixed probability of success, and thus, the time to successfully transmit a packet can be modeled as a geometric random variable. In this setting the authors show that the LDF policy is ``optimal.'' 
However, the results are restricted to a single resource shared by users with geometric workloads. In this paper we study the performance of LDF in a more general setting which includes this prior work as a special case. 
This initial set of papers motivated follow-up work in wireless context, see e.g., \cite{HoK14,MLH10,JaS11}.

The performance of LDF and similar policies has also been studied in \cite{MCK95,DiW06,JLS07,KWJ13}. The authors in \cite{DiW06} consider the generalized switch model and were the first to propose the notion of ``local pooling'' as a sufficient condition for the Longest Queue First (LQF) policy to be throughput optimal. 
Subsequently, the work in \cite{JLS07} considers a multi-hop wireless network under a node-exclusive interference model and shows that the efficiency ratio of the greedy maximal matching policy, which is essentially LQF, equals to the ``local pooling'' factor of the network graph. More recently, the authors in \cite{KWJ13} consider real-time traffic in ad hoc wireless networks under a link-interference model and also characterize the efficiency ratio of the LDF policy. 

The results in \cite{DiW06,JLS07,KWJ13} depend on the constant service rate model and the specific interference model, i.e., where the set of links/queues that can be scheduled simultaneously is restricted. 
These models may be appropriate in some wireless/queueing networks but do not necessarily hold in our broader context, e.g., soft real-time applications with stochastic workloads. Also, \cite{DiW06} lacks a performance analysis of LQF when it is not optimal and the works in \cite{JLS07,KWJ13} focus on the efficiency ratio of LDF-like policies but lack a characterization of the full capacity region of these policies. 
Moreover, when the system can deliver more than the requirements, either the QoS requirements for real-time traffic or the throughput requirements for queueing systems, there is no discussion of how to manage the allocation of the ``excess capacity'' across users. 

The authors in \cite{TaE92,TaE93,MMA99,DaP00,Sto04} propose max weight 
scheduling policies for different types of queueing systems and show them to be throughput optimal 
via the approaches summarized in \cite{DoM94,DoM97,MeT08}. 
\extendedStart
The authors in \cite{Nee09} and \cite{VBY13} further 
characterize the delay of the max weight policy, and study its inefficiency in spatial wireless networks, respectively. 
\commentEnd\fi
As we will see in the sequel we too discuss a max weight-like scheduling policy, but it suffers from the usual complexity problems when the decision space
is large and it requires excessive amounts of information, motivating us to consider simpler policies. 

Additional related work includes work on modeling and scheduling of real-time tasks, see e.g., \cite{SAA04,DaB11,LiI08,ShS02}. 

{\bf \em Our Contributions. }
In this paper, we contribute to the theoretical understanding and performance characterization of the Largest Deficit First (LDF) policy with applications to resource allocation to support real-time services. We make three key contributions. 

First, we propose a novel general model for a class of systems supporting priority-based resource allocation 
and study different dynamic prioritization policies. This model is general in terms of the ``impact'' the priority decisions can have on the QoS payoffs. Specifically, in each period the payoffs under a priority decision are modeled by a random vector, which includes as special cases the single resource model, the geometric/constant workload and/or specific interference model adopted in prior work. 
For this general model, we propose a general inner bound $R_\text{IB}$ for the QoS feasibility region of LDF prioritization policy. 

Second, with an additional property, {\em monotonicity in payoffs}, we characterize the geometry of the inner bound $R_\text{IB}$. 
Based on this, we further propose a sufficient condition for the optimality of the LDF policy and characterize the efficiency ratio of LDF. 
In practice, understanding the geometry of $R_\text{IB}$ enables us to understand and identify possible bottlenecks in the priority-based resource allocation infrastructure. 
We also show that the LDF policies (as well as a hierarchical-LDF version) are optimal when there are two classes of exchangeable users. 

Finally, we also consider the class of weighted LDF policies, which enable us to explore the allocation of ``excess payoffs'' when the system has ``excess'' capacity. Simulation results are exhibited to show the impact of weights and to characterize the clustering of failures. 

{\bf \em Paper Organization. }
The paper is organized as follows: Section 2 introduces our general model for systems supporting priority-based resource allocation. Section 3 develops theoretical results and characterizes the performance of the weighted LDF policies while Section 4 presents some examples for the optimality of the weighted LDF/hierarchical-LDF policies. Section 5 discusses some practical issues while the impact of weights is exhibited via simulation in Section 6. Section 7 concludes the paper and points to future work. Some of the proofs are provided in the Appendix.

\section{System Model}
\label{sec_system_model}
We consider applications which periodically generate random workloads with the same period and specify long-term QoS requirements. In the sequel we let a user denote a specific instance of such an application. 

We begin by introducing a general model for systems that allocate resources in each period based on the following decomposition: 
(1) users are assigned priorities dynamically, e.g., at runtime, according to a function of the past history, and (2) the system allocates resources based on these priorities. 

For the most part in this paper, the manner in which (2) is carried out will not be our concern. Instead our focus will be on how to perform dynamic user prioritization to achieve optimal (or near-optimal) system performance when combined with a given underlying mechanism for (2). 
\extendedStart
In our follow-up work \cite{DuD16S}, we consider a specific system model and study the combined design of (1) and (2). 
\commentEnd\fi

\subsection{General Model for Systems Supporting Priority-Based Resource Allocation} 
\label{subsection_general_payoff_model}
We consider an abstract system that serves $n$ users indexed from $1$ to $n$. Let $\fullUserSet = \{1, 2, \cdots, n\}$ be the user set. The system operates in discrete time, over periods $t = 1, 2, \cdots$. In each period, it picks a user {\em priority decision} $\mathbf{d} = (d_1, d_2, \dots, d_n)$ where $d_m$ is the index of the user with $m^{\text{th}}$ highest priority. We let $D$ denote the set of all possible priority decisions and let $|D|$ represent the number of possible decisions, thus, $|D| = n!$ 

In each period, given the priority decision $\mathbf{d}$ passed to the underlying resources, since there are intrinsic uncertainties in users' workloads, 
each user $i$ achieves a non-negative random QoS payoff, denoted by $V_i(\mathbf{d})$. We let $\mathbf{V}(\mathbf{d}) = (V_1(\mathbf{d}), V_2(\mathbf{d}), \cdots, V_n(\mathbf{d}))$. 
We assume the payoffs are independent across periods. 
The distribution of $\mathbf{V}(\mathbf{d})$ depends on the selected priority decision $\mathbf{d}$ and the expected payoff vector given $\mathbf{d} \in D$ is denoted by $\mathbf{p}(\mathbf{d}) = \expectation[\mathbf{V}(\mathbf{d})]$. We assume all possible payoff vectors form a finite rational set. 
\extendedStart
Moreover, we naturally assume that for each user $i\in \fullUserSet$, there exists a decision $\mathbf{d}$ such that $p_i(\mathbf{d}) > 0$. 
\commentEnd\fi

Each user requires a long-term average QoS payoff $\reqscalar_i \geq 0$ as the QoS requirement. We let $\reqvec = (\reqscalar_1, \reqscalar_2, \cdots, \reqscalar_n)$ and assume $\reqscalar_i$'s are rational\footnote{All the results in this paper can be generalized to models with irrational values. For simplicity in the proof we do not consider that level of generality. }. We denote by $\mathbf{d}(t)$ the priority decision at period $t$. To keep track of the deficits between required and achieved QoS payoffs, for each user $i \in \fullUserSet$ and period $t+1$, we define\footnote{We truncate the deficit at $0$ for the convenience of defining feasibility in the sequel. Removing the truncation won't change the results in the paper. }
\begin{align}
\label{align_deficit_in_general_payoff}
X_i(t+1) = [X_i(t) + \reqscalar_i - V_i(\mathbf{d}(t+1))]^+, 
\end{align}
where $[x]^+ = \max[x, 0]$. 

The goal is thus to devise user prioritization policies which will meet users' long-term payoff requirements. 

\begin{definition}
A {\bf user prioritization policy} is a stationary policy that picks a priority decision $\mathbf{d}(t+1) \in D$ at period $t+1$ based on the following:
\begin{longitem}
\item[-] users' payoff requirement vector $\reqvec$; 
\item[-] expected payoff vectors $P = \{\mathbf{p}(\mathbf{d}) | \mathbf{d} \in D\}$; 
\item[-] and, the deficits $\mathbf{X}(t) = (X_1(t), X_2(t), \cdots, X_n(t))$. 
\end{longitem}
\end{definition}

The process $\{\mathbf{X}(t)\}_{t\geq 1}$ is a Markov chain under any such policy. We assume the initial state $\mathbf{X}(0)$, the requirements $\reqvec$, the set of all possible payoff vectors and the user prioritization policy make $\{\mathbf{X}(t)\}_{t\geq 1}$ an irreducible Markov chain. 

\begin{definition}
\label{defn_feasibility_pr}
A payoff requirement vector $\reqvec$ is said to be {\bf feasible} if there exists a user prioritization policy $\eta$ under which the Markov chain $\{\mathbf{X}(t)\}_{t\geq 1}$ is positive recurrent. We also say this policy fulfills this requirement vector. 
\end{definition}

The expected payoff vectors $P = \{\mathbf{p}(\mathbf{d}) | \mathbf{d} \in D\}$ could in principle be statistically inferred from the history events or by repeated experiments. 
However, in a practical setting this can be challenging and it is of interest to find a policy that performs well and uses little a-priori information regarding the exponential set of expected payoff vectors $P$. 

Note that this model is general in the sense that the ``impact'' of priority decisions $\mathbf{d} \in D$ on the QoS payoff vectors $P$ is at this point general, whereas the specific resource and workload models in prior work, e.g., \cite{DiW06,JLS07,KWJ13}, implicitly impose properties on $P$ and therefore restrict the results significantly. 

\subsection{Example: Centralized Computing System for Real-Time Applications}
\label{subsection_SRT_model}
Our model can for example capture a centralized computing infrastructure supporting Soft Real-Time (SRT) applications where the $n$ users share compute resources. In a cloud-based collaborative video conferencing context, a user might correspond to an individual end user and the period length might correspond to the length of a group of video frames. 

The users generate streams of tasks periodically. Specifically in each period a user generates several tasks. A task may further consist of a graph of possibly dependent sub-tasks with (possibly) random processing requirements, i.e., workloads. These tasks/sub-tasks need to be fully completed before the end of the period. 
For real-time services, it is generally useless to process a task after its deadline. For example, in the video conferencing context it is not desirable to present an out-of-date frame.  Therefore, we assume tasks/sub-tasks not completed on time are dropped. 

In each period $t$, the user prioritization policy picks a user priority decision $\mathbf{d}(t)$, based on which compute resources are allocated to process tasks. 
Given the task processing results, a payoff $V_i(\mathbf{d}(t))$ is achieved for each user $i$ based on whether the tasks were successfully processed, or how much of the task graphs were completed. 
In general, $V_i(\mathbf{d}(t))$ may represent any user-specific QoS payoff per period, that can be averaged over time, e.g., the quality/resolution of video frame processing, or the number of task completions. Accordingly the vector $\reqvec$ represents the long-term average QoS requirements. 

\dissertationStart
Note that these payoff vectors depend heavily on the internal design of the system and in this paper we assume the internal design is fixed and has been carefully optimized by the computing system architect. 
In practice, it may be challenging to jointly optimize the internal design and the stationary priority scheduling policy to provide the best QoS to users. 
A discussion of different internal designs in a specific computing system model is provided in [?]. 
\commentEnd\fi

\subsection{Example: Complex Networks and Flexible Modeling of Application Execution Payoffs}
As indicated in the introduction, our model also applies to a complex network of heterogeneous compute and communication resources, as long as users periodically and synchronously generate tasks that require timely processing on diverse resources and moving around in the network, e.g., as shown in Figure~\ref{fig_complex_network_resources}. 

Given the priority decision in each period, the network of resources coordinate according to some priority-based resource allocation mechanism to accelerate the processing of tasks with high priorities, by reducing the communication/queueing delays, processing with higher processor speed, allocating more shared resources, etc. 

Again, different users can define their payoffs in different ways and specify their QoS requirements accordingly. 

\dissertationStart
\myComments{In dissertaion, we have to re-organize this section. }
\subsection{Other Interpretations of General Payoff Model}
Our general payoff model applies beyond the soft real-time context. For example, for each user $i\in \fullUserSet$, by viewing $\reqscalar_i$ as expected arrival rates and $V_i(\mathbf{d}(t+1))$ as the service rate given priority decision $\mathbf{d}(t+1)$ is selected in period $t+1$, the general payoff model captures the queueing system, e.g., the generalized switch model or the wireless networks considered in prior work \cite{DiW06,JLS07,KWJ13}. 

However, our general payoff model is general in the sense that the ``impact'' of priority decisions $\mathbf{d} \in D$ on the payoff vectors $P$ is left unspecified, while the interference model in these prior work implicitly imply properties on P. 

In the next section, we first derive results in the general payoff model and then add further properties on $P$ to get cleaner results. 
\commentEnd\fi

\section{Performance Analysis}

In this section we shall develop theoretical results for such systems. 
\infocomStart
To save space we have deferred proofs of these results to the extended version of this paper available at \cite{EXT2}. 
\commentEnd\fi
Some of these results are similar to prior work but in the more general model while other results are completely new. 
For completeness we shall develop a self-contained theoretical framework. 
\dissertationStart
In the sequel we provide examples to help understand these results in the context of SRT applications. 
\commentEnd\fi

\subsection{System Feasibility Region and Feasibility Optimal Policy}
\label{subsection_feasibility_region_and_optimal_scheduling}

The set of all feasible long-term payoff requirement vectors will be referred to as the {\em system feasibility region} $\feasibilityRegion$. 
We let $\feasibilityRegion_\eta$ denote the feasibility region of a user prioritization policy $\eta$. 
To characterize $\feasibilityRegion$ we introduce some further notation. 

A vector $\mathbf{x}$ is said to be {\em dominated} by a vector $\mathbf{y}$ if $x_i \leq y_i$ for all $i$ and is denoted by $\mathbf{x} \preceq \mathbf{y}$. We define $\mathbf{x} \prec \mathbf{y}$, $\mathbf{x} \succeq \mathbf{y}$ and $\mathbf{x} \succ \mathbf{y}$ in a similar manner. 

Given the set of priority decisions $D$ and the expected payoff vectors $P = \{\mathbf{p}(\mathbf{d}) | \mathbf{d} \in D \}$, we let $C$ be the set of requirement vectors $\reqvec \in \mathbb R^n_{+}$ which are dominated by a vector in the convex hull of $P$ denoted Conv($P$), i.e., 
\begin{align}
\label{align_C}
C \equiv \{ \reqvec \in \mathbb R^n_{+} ~|~ \exists \mathbf{x} \in \text{Conv($P$)} \text{ such that } \reqvec \preceq \mathbf{x} \}. 
\end{align}

Figure~{\ref{fig_feasible_region}} exhibits $C$ for a two-user (left figure) and three-user (right figure) setting. 
In the two-user setting, the points labeled $\mathbf{p}(\mathbf{d}_1)$ and $\mathbf{p}(\mathbf{d}_2)$ are the expected payoff vectors of two priority decisions, i.e., where User 1 or User 2 has higher priority, respectively. The shadowed area represents $C$. 
\dissertationStart
Clearly these two expected payoff vectors are always on a line. 
\commentEnd\fi
In the three-user setting, the circles represent the $6$ possible expected payoff vectors, and the region dominated by their convex hull is $C$. 
Note that in a $n$-user scenario where $n\geq 3$, as displayed the expected payoff vectors need not be on a hyperplane in the $n$-dimensional space. 
As we will see this is essentially the source of complexity in studying such systems. 

\begin{figure}[htp]
  \centering
  \includegraphics[width=0.65\textwidth]{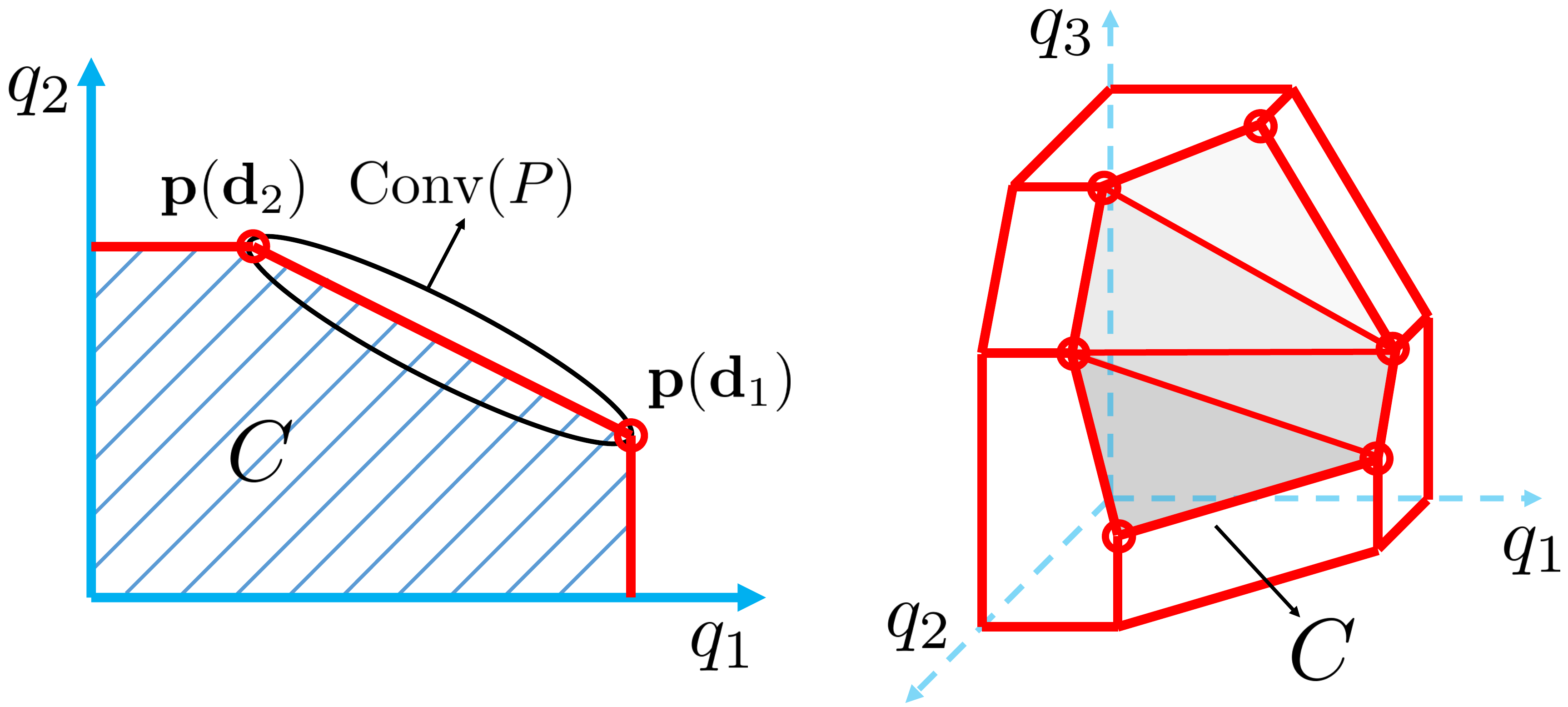}
  \caption{Examples of set $C$ when $n=2$ and $n=3$. }
  \label{fig_feasible_region}
\end{figure}

Clearly, for any requirement vector $\reqvec$ in the interior of $C$, denoted by $\text{int}(C)$, one can achieve $\reqvec$ if one is allowed to do probabilistic time sharing among priority decisions by picking decisions according to a pre-computed probability distribution whose mean payoff dominates $\reqvec$. Therefore, $\text{int}(C) \subseteq \feasibilityRegion$. We can also show the following result. 

\begin{lemma}
\label{lemma_R_in_close_C}
The system feasibility region $F$ is such that
$$
\feasibilityRegion \subseteq \textnormal{cl}(C), 
$$
where $\textnormal{cl}(C)$ is the closure of $C$. 
\end{lemma}

Intuitively, if $\reqvec$ is feasible, it is fulfilled by some user prioritization policy that in the long-term picks each priority decision some fraction of the time and thus, $\reqvec$ is dominated by some point in the convex hull of $P$. This is similar to prior work, e.g., \cite{TaE92}. 
\infocomStart
See the extended version of this paper \cite{EXT2} for the proof. 
\commentEnd\fi
\extendedStart
See the appendix for a detailed proof. 
\commentEnd\fi
In other words, $C$ is different from $\feasibilityRegion$ by at most a boundary, and therefore, characterizes $\feasibilityRegion$ for practical purposes. Thus, in the sequel we will also refer to $C$ as the system feasibility region. 

Ideally, it is desirable to devise an ``optimal'' policy that can fulfill all feasible requirements. 
More formally, a user prioritization policy $\eta$ is said to be {\em feasibility optimal} if $\text{int}(C) \subseteq \feasibilityRegion_{\eta} \subseteq \text{cl}(C)$. 
Similar to prior work \cite{TaE92,TaE93}, the following max weight-like policy is one such feasibility optimal policy. 

\begin{definition}
The {\bf deficit-based max weight (MW)} prioritization policy is such that, at period $t+1$, given the deficit vector $\mathbf{X}(t)$ computed by (\ref{align_deficit_in_general_payoff}), it picks a priority decision $\mathbf{d}(t+1)$ that satisfies
\begin{align}
\label{align_max_weight}
\mathbf{d}(t+1) \in \argmax\limits_{\mathbf{d} \in D}\langle \mathbf{X}(t), \mathbf{p}(\mathbf{d}) \rangle, 
\end{align}
where $\langle \mathbf{x}, \mathbf{y}\rangle$ is the inner product of two vectors. 
\end{definition}

\begin{theorem}
\label{theorem_MW_feasibility_optimal}
\infocomStart
The feasibility region of the MW policy $\feasibilityRegion_{\text{MW}}$ is such that 
$$
\text{int}(C) \subseteq \feasibilityRegion_{MW} \subseteq \text{cl}(C), 
$$
and therefore, the MW policy is feasibility optimal. 
\commentEnd\fi
\extendedStart
The system feasibility region $\feasibilityRegion$ and the feasibility region of the MW policy $\feasibilityRegion_{\textnormal{MW}}$ are related to $C$ as follows, 
$$
\textnormal{int}(C) \subseteq \feasibilityRegion_{\textnormal{MW}} \subseteq \feasibilityRegion \subseteq \textnormal{cl}(C), 
$$
and therefore, the MW policy is feasibility optimal. 
\commentEnd\fi
\end{theorem}

\noindent See Appendix \ref{appendix_pf_thm_MW_feasibility_optimal} for the proof.

However, the MW policy and time sharing policies require full knowledge of $P$ which is challenging in complex practical systems. Moreover, these policies are hard to implement since they involve solving fairly complex optimization problems, i.e., Eq (\ref{align_max_weight}). Changes in the user set or payoff requirement vector $\reqvec$ will also impact the realization of these policies. 
In summary, the requirements in terms of a-priori knowledge, the computational complexity and lack of flexibility to changes make them hard to use in practice. This motivates the policies considered in the next subsection. 

For ease of reference, Table \ref{tab_notation} provides a summary of the notation used to denote various regions used in the rest of the paper\textemdash some of these are introduced in the sequel. 

\begin{table}[h]
\normalsize
    \tbl{Notation of regions. \label{tab_notation}}{
    \centering
    \begin{tabular}{| c || l |}
    \hline
    {\bf Regions} & {\bf Description}\\
    \hline
    \hline
    $\feasibilityRegion$ & System Feasibility Region. \\
    \hline
    Conv$(P)$ & Convex hull of the expected payoff vectors.     \\
    \hline
    $C$ & Region dominated by Conv$(P)$. \\ 
    \hline
    $\feasibilityRegion_{\mathbf{w}\text{-LDF}}$ & Feasibility region of the $\mathbf{w}$-LDF policy. \\
    \hline
    $R_\text{IB}$ & An inner bound for $\feasibilityRegion_{\mathbf{w}\text{-LDF}}$\\
    \hline
    $B$ & Dominant of the convex hull. \\
    \hline
    $R$ & Region characterizing the geometry of $R_\text{IB}$. \\
    \hline
    \end{tabular}
    }
\end{table}

\subsection{Weighted LDF Policies and Associated Feasibility Regions}
\label{subsection_LDF}
The LDF user prioritization policies require no a-priori knowledge of the system, are simple to implement and adapt easily to changes in $\reqvec$ or the user set. In particular we shall characterize the feasibility regions of these policies by providing an inner bound. 

\begin{definition}
\label{defn_w_LDF}
Given a vector $\mathbf{w} = (w_1, w_2, \cdots, w_n) \succ \mathbf{0}$, the {\bf weighted Largest Deficit First ($\mathbf{w}$-LDF)} user prioritization policy is such that, at period $t+1$, given the deficit vector $\mathbf{X}(t)$, it picks a priority decision $\mathbf{d}$ that satisfies
$$
w_{d_1}X_{d_1}(t) \geq w_{d_2}X_{d_2}(t) \geq \cdots \geq w_{d_n}X_{d_n}(t),
$$
with ties broken arbitrarily (possibly randomly). In other words, it sorts the weighted deficits of users and assigns priorities accordingly. 
Let $\mathbf{1} \equiv (1, 1, \cdots, 1)$. 
We refer to the policy with $\mathbf{w} = \mathbf{1}$ the {\bf Largest Deficit First (LDF)} policy. 
\end{definition}

\dissertationStart
The LDF policy has been explored before \cite{DiW06, JLS07, KWJ13} in the context of specific resource, workload and payoff models. We also generalize it to a class of policies with different weight vectors and provide better geometric characterization of the feasibility region. 
\commentEnd\fi

Clearly, the $\mathbf{w}$-LDF prioritization policies do not require knowledge of the expected payoff vectors $P$.
Note that we still use deficit feedback to stabilize the system. 
In terms of computational complexity, solving (\ref{align_max_weight}) is $O(n!)$ while sorting weighted deficits only requires $O(n\log n)$. 
It also allows us to further differentiate the performance across users by assigning different weights. The impact of weights is discussed in Section \ref{sec_simulations}. 

Prior work has established that the LDF policy need not be feasibility optimal. Therefore, a key question is whether the feasibility regions for the $\mathbf{w}$-LDF policies are acceptable and to characterize the gap between their feasibility regions and the system feasibility region $\feasibilityRegion$. 
To that end, we first provide an inner bound, denoted by $R_{\text{IB}}$, for the feasibility region of any $\mathbf{w}$-LDF policy. 

\begin{theorem}
\label{thm_R_IB}
For any $\mathbf{w} \succ \mathbf{0}$, an inner bound for the feasibility region of the $\mathbf{w}$-LDF policy $\feasibilityRegion_{\mathbf{w}\textnormal{-LDF}}$ is given by
$
\textnormal{int}(R_{\textnormal{IB}}) \subseteq \feasibilityRegion_{\mathbf{w}\textnormal{-LDF}}, 
$ 
where
\begin{align}
\label{align_R_IB_defn}
R_{\textnormal{IB}} \equiv \{\reqvec \in \mathbb R^n_{+} ~|~ \exists \ribvec \succ \mathbf{0} \text{ such that } \forall S\subseteq \fullUserSet, 
\sum\limits_{i\in S} \ribscalar_i {\reqscalar}_i \leq \min\limits_{\mathbf{d} \in D(S)} \sum\limits_{i\in S} {\ribscalar}_i p_i(\mathbf{d})
\}
\end{align}
where $D(S)$ denotes the set of all priority decisions that assign the highest $|S|$ priorities to users in $S$. 
\end{theorem}

In other words, if $\reqvec\in R_{\text{IB}}$, it is feasible under all $\mathbf{w}$-LDF policies except perhaps boundary points. 
The underlying intuition for this bound is as follows. A vector $\reqvec$ is in $R_\text{IB}$ if
there is a weight vector $\ribvec \succ \mathbf{0}$ such that for any subset of users $S$, and decisions giving users in $S$ the highest priorities, the weighted sum of payoff requirement $\sum\limits_{i \in S}\ribscalar_i \reqscalar_i$ will not exceed the least sum weighted payoff $\sum\limits_{i\in S} {\ribscalar}_i p_i(\mathbf{d})$. 
Based on $\ribvec$, we can construct an appropriate Lyapunov function to show feasibility for $\reqvec$ and each $\mathbf{w}$. 
\infocomStart
See the extended version of this paper \cite{EXT2} for the proof. 
\commentEnd\fi
\extendedStart
See Appendix \ref{appendix_pf_thm_R_IB_chap2} for the proof. 
\commentEnd\fi

\dissertationStart
In other words, $\reqvec\in R_{\text{IB}}$ and is feasible under the $\mathbf{w}$-LDF policy if there is a weight vector $\ribvec \succ \mathbf{0}$ such that for any subset of users $S$, by giving users in $S$ the highest priorities, the weighted sum of payoff requirement $\sum\limits_{i \in S}\ribscalar_i \reqscalar_i$ will not exceed the least sum weighted payoff $\sum\limits_{i\in S} {\ribscalar}_i p_i(\mathbf{d})$. 
\commentEnd\fi

Understanding the geometry of $R_\text{IB}$ enables us to characterize the performance gap between $\mathbf{w}$-LDF and feasibility optimal policies. 
Let us informally consider the geometry of $R_\text{IB}$ for the two special cases in Figure~{\ref{fig_feasible_region}}. 
In the two-user case in Figure~{\ref{fig_feasible_region}}, $R_\text{IB}$ is the same as $C$ and thus, the $\mathbf{w}$-LDF policies are feasibility optimal. 
However, in the three-user case in Figure~{\ref{fig_feasible_region}}, this need not be true. 
Indeed, in this setting, the region $R_\text{IB}$ corresponds to $C$ minus the convex hull of $P$, modulo some boundary points. 
This is exhibited in Figure~{\ref{fig_three_user_R_IB}}. 
In the next subsection, we will formalize these observations and show under what conditions they hold true. 

\begin{figure}[htp]
  \centering
  \includegraphics[width=0.8\textwidth]{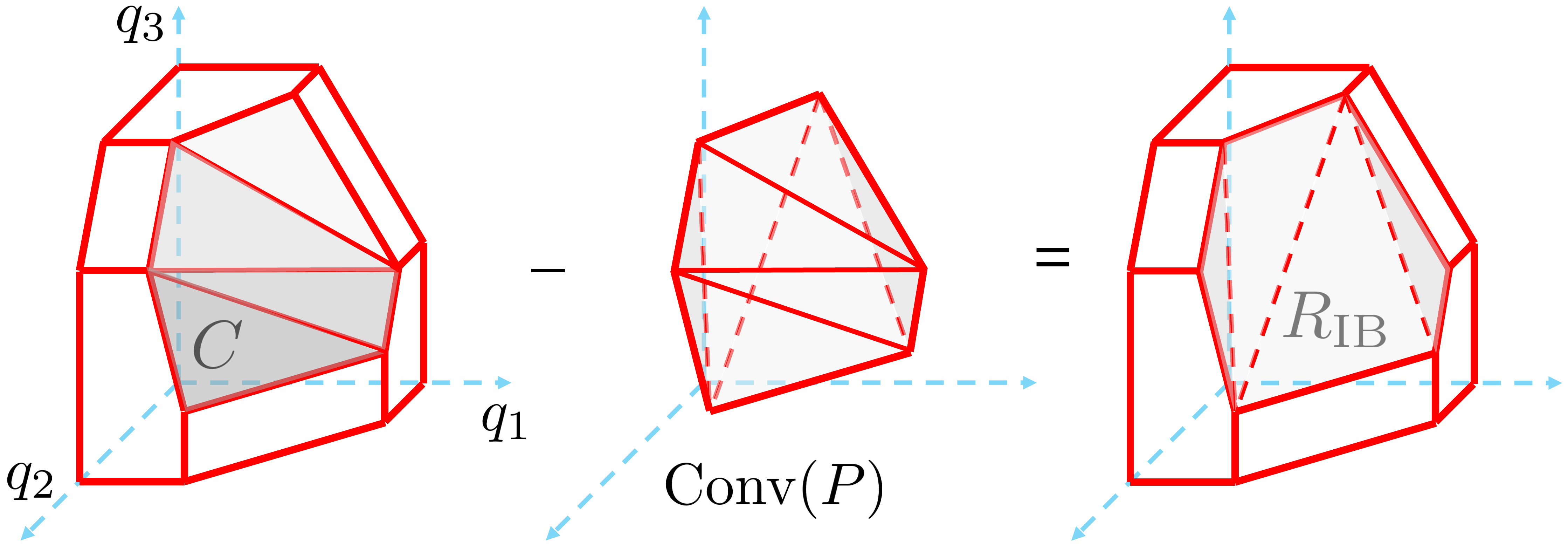}
  \caption{Visualizing $R_{\text{IB}}$ for the three-user scenario in Figure~{\ref{fig_feasible_region}}. \add{In this example, the expected payoff vectors $P = \{\mathbf{p}(\mathbf{d}) | \mathbf{d} \in D \}$ are not on the same hyperplane} and $R_{\text{IB}} = \text{cl}(C - \text{Conv}(P))$. }
  \label{fig_three_user_R_IB}
\end{figure}

\subsection{Geometry of $R_{\text{IB}}$ under Monotonicity in Payoffs}
\label{subsection_characterizing_R_IB}

In order to formally characterize the geometry of $R_{\text{IB}}$ we will add a further natural requirement to the general model. 

We define $S_i(\mathbf{d})$ to be the set of users that have higher priorities than user $i$ under decision $\mathbf{d}$. 

\begin{definition}
The system with expected payoff vectors $P = \{\mathbf{p}(\mathbf{d}) | \mathbf{d} \in D \}$ is said to satisfy {\bf \priorityImplication} if, for any two priority decisions $\mathbf{d}_1$ and $\mathbf{d}_2$ and any user $i$ such that $S_i(\mathbf{d}_1) \subseteq S_i(\mathbf{d}_2)$, it is true that $p_i(\mathbf{d}_1) \geq p_i(\mathbf{d}_2)$. We call this {\bf monotonicity in payoffs} for short. 
\end{definition}

In other words, a user $i$ can expect to get a higher payoff if some users with higher priority are re-assigned lower priorities. 
\dissertationStart
Note we are not comparing the expected payoffs of different users since payoffs can be defined in different ways for different users and may not be comparable. 
\commentEnd\fi
This property characterizes in a broad sense how priorities impact the expected payoffs when the underlying system allocates resources. 
It is a natural condition but need not hold in general. 

We shall define $B$ to be the set of payoff requirement vectors $\reqvec$ which dominate a vector in the convex hull of $P$, i.e., 
$$
B \equiv \{ \reqvec \in \mathbb R_+^n ~|~ \exists \mathbf{x} \in \text{Conv($P$)} \text{ such that } \reqvec \succeq \mathbf{x} \}. 
$$
We call $B$ the {\em dominant of the convex hull}. Contrast this to the definition of $C$ in (\ref{align_C}). 

For the special cases in Figure~{\ref{fig_feasible_region}} and \ref{fig_three_user_R_IB}, $B \cap C$ equals to $\text{Conv}(P)$, but in general it can be larger than $\text{Conv}(P)$. Figure~{\ref{fig_eg_for_B_intersect_C}} shows a conceptual picture of what could happen. The three circles represent three possible expected payoff vectors. Here, the whole shadowed area $B \cap C$ is larger than the region Conv$(P)$ which is the triangle formed by the three circles. 
Note that this is only a conceptual example to help visualize $B\cap C$ in higher dimensions. In reality for two dimensions, i.e., systems with two users, we know there are only $2$ expected payoff vectors as shown in Figure~{\ref{fig_feasible_region}}. 

\begin{figure}[htp]
  \centering
  \includegraphics[width=0.3\textwidth]{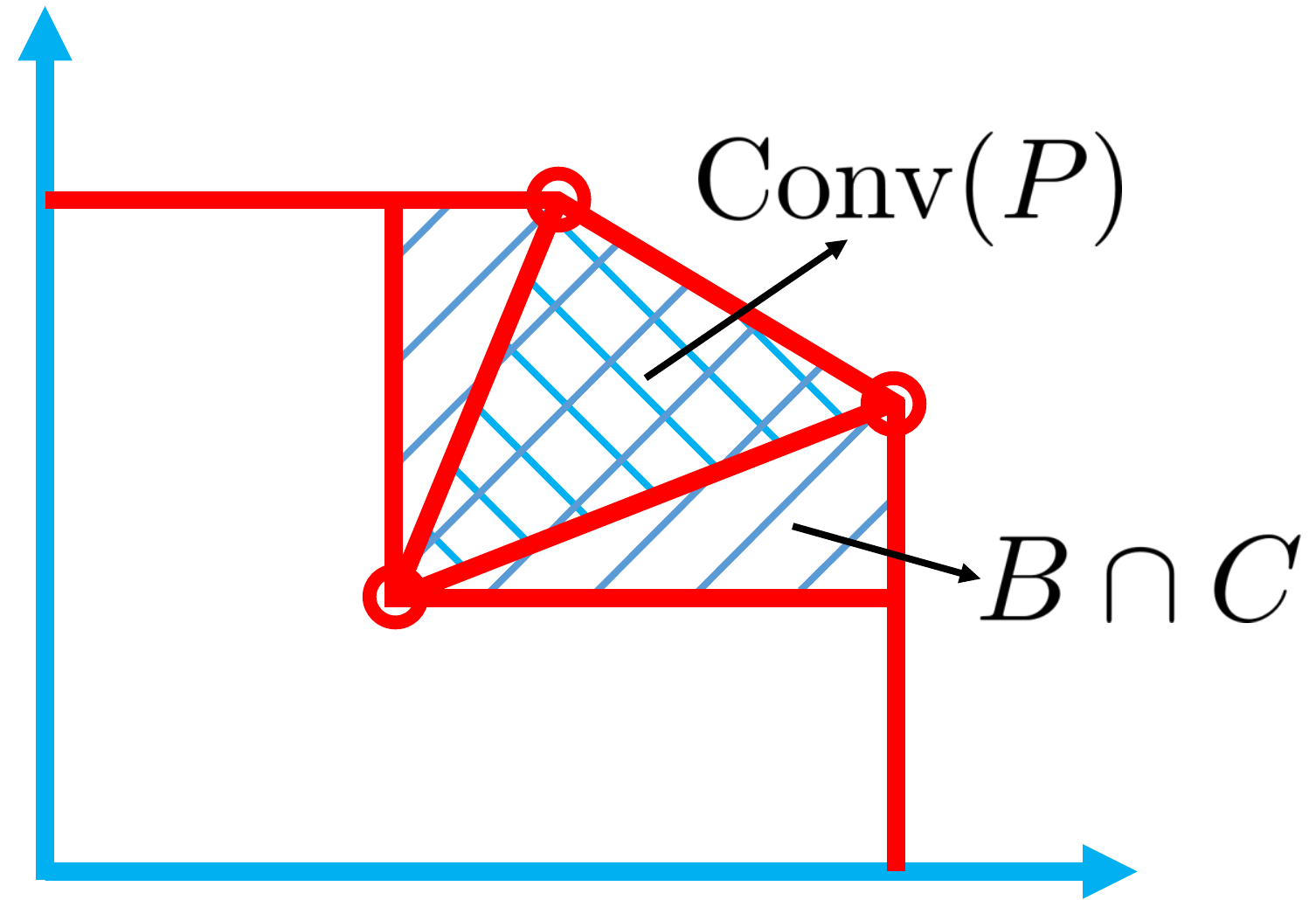}
  \caption{An example where $B\cap C$ is larger than Conv$(P)$. }
  \label{fig_eg_for_B_intersect_C}
\end{figure} 

In the sequel we will see that given monotonicity in payoffs, $R_{\text{IB}}$ is obtained by ``removing'' $B \cap C$, rather than just $\text{Conv}(P)$ from $C$. 
To develop this result we need some further notation associated with each subset of users $S\subseteq \{1, 2, \cdots, n\}$. 

The projection of a vector $\mathbf{x}$ on the subspace of $S$ is denoted by $\mathbf{x}^S$, i.e., 
$$
x_i^S = \left\{ 
   \begin{array}{l l}
     x_i & \quad \text{if $i \in S$}	\\
     0 & \quad \text{otherwise}.	\\
   \end{array} \right.
$$


We let $P^S \equiv \{\mathbf{p}^S(\mathbf{d}) | \mathbf{d} \in D(S)\}$ represent the projections of expected payoff vectors corresponding to decisions in $D(S)$, i.e., which assign the highest priorities to users in $S$. 

Given a subset $S$ and $P^S$, we define the feasibility region $C^S$ and the dominant of the convex hull $B^S$ as follows. 
\begin{align*}
C^S \equiv \{& \reqvec^S \in \mathbb R^n_+ ~|~ \exists \mathbf{x}^S \in \text{Conv($P^S$)} \text{ such that } \reqvec^S \preceq \mathbf{x}^S \}, 	\\
B^S \equiv \{& \reqvec^S \in \mathbb R^n_+ ~|~ \exists \mathbf{x}^S \in \text{Conv($P^S$)} \text{ such that } \reqvec^S \succeq \mathbf{x}^S \}. 
\end{align*}

Note that $C^S$ and $B^S$ are not necessarily the same as projecting $C$ and $B$ on the subspace of $S$, respectively. This is because in the definitions of $C^S$ and $B^S$, 
we only focus on a subset of decisions $D(S)$ rather that the full decision set $D$. 
\dissertationStart
we only focus on priority decisions $\mathbf{d} \in D(S)$ rather that the full decision set $D$. 
\commentEnd\fi

Let us now define a region $\concaveHull$ which will help characterize the geometry of the inner bound $R_{\text{IB}}$. 

\begin{definition}
\label{defn_concave_hull}
Let $\concaveHull$ be defined as follows:
$$
\concaveHull \equiv \{ \reqvec \in \mathbb R^n_+ ~|~  \forall S \subseteq \fullUserSet, \reqvec^S \in  C^S \setminus B^S \},
$$
where $C^S \setminus B^S = \{\reqvec^S|\reqvec^S \in C^S, \reqvec \notin B^S \}$. 
In other words, any $\reqvec \in R$ satisfies that for any user subset $S$, its projection on the subspace of $S$ belongs to the set $C^S \setminus B^S$, which is the feasibility region $C^S$ minus the dominant of the convex hull $B^S$. 
\end{definition}

One can visualize obtaining the set $\concaveHull$ as a process of removing $B^S \cap C^S$ from $C^S$ in all subspaces corresponding to all subsets $S$. The geometry of $R_\text{IB}$ is then captured as follows. 
\begin{theorem}
\label{thm_visualize_R_IB}
If the system satisfies monotonicity in payoffs, then the inner bound region $R_\textnormal{IB}$ is such that
$$
\textnormal{int}(\concaveHull) \subseteq R_\textnormal{IB} \subseteq \textnormal{cl}(\concaveHull). 
$$
\end{theorem}

\infocomStart
See the extended version of this paper \cite{EXT2} for this somewhat technical proof. 
\commentEnd\fi
\extendedStart
\noindent See Appendix \ref{appendix_pf_thm_visualize_R_IB} for this somewhat intricate argument. 
\commentEnd\fi

\subsection{Sufficient Condition for $\mathbf{w}$-LDF's Optimality}

By Theorem \ref{thm_R_IB} and Theorem \ref{thm_visualize_R_IB}, we immediately get
\begin{align}
\label{align_A_vs_R_LDF}
\text{int}(\concaveHull) \subseteq \text{int}(R_\text{IB}) \subseteq \feasibilityRegion_{\mathbf{w}\text{-LDF}}. 
\end{align}
Since $\concaveHull$ is obtained by removing $B^S \cap C^S$ from $C^S$ for each $S$, if what is removed is nothing more than a boundary, the difference between $\concaveHull$ and $C$ is at most a boundary and thus $\mathbf{w}$-LDF policies are feasibility optimal. It is easy to see this happens when vectors in $P^S$ lie on a hyperplane for each subset of users $S$. This can be formalized as follows. 

\begin{definition}
\label{defn_subset_payoff_equivalence}
The system with expected payoff vectors $P = \{\mathbf{p}(\mathbf{d}) | \mathbf{d} \in D \}$ is said to satisfy {\bf subset payoff equivalence} if for each subset of users $S$ the vectors in $P^S = \{\mathbf{p}^S(\mathbf{d}) | \mathbf{d} \in D(S)\}$ lie on a hyperplane, i.e., there exists a nonzero $\boldsymbol{\alpha}^S \succeq \mathbf{0}$ such that for all $\mathbf{d}_1, \mathbf{d}_2 \in D(S)$, 
$$
\langle \boldsymbol{\alpha}^S, \mathbf{p}^S(\mathbf{d}_1) \rangle = \langle \boldsymbol{\alpha}^S, \mathbf{p}^S(\mathbf{d}_2) \rangle. 
$$
\end{definition}

\begin{theorem}
\label{thm_sufficient_condition_for_LDF_optimality}
If the system satisfies monotonicity in payoffs and subset payoff equivalence, then 
$$
\textnormal{int}(C) \subseteq \feasibilityRegion_{\mathbf{w}\textnormal{-LDF}} \subseteq \textnormal{cl}(C), 
$$
and therefore, the $\mathbf{w}$-LDF policies are feasibility optimal. 
\end{theorem}
\extendedStart
\noindent Please refer to Appendix \ref{appendix_pf_thm_sufficient_condition_for_LDF_optimality} for detailed proof. 
\commentEnd\fi

The conditions for this theorem are akin but not equivalent to the conditions introduced in \cite{DiW06} for the generalized switch model. 
Specifically, we require the system to satisfy monotonicity in payoffs and subset payoff equivalence. The work in \cite{DiW06} requires local pooling in the generalized switch model. 
In the model in \cite{DiW06}, given a priority decision $\mathbf{d} = (d_1, d_2, \cdots, d_n)$ where $d_i$ is the index of the queue with the the $i^\text{th}$ highest priority, the queue service rate vector can be denoted by $\mathbf{m}(\mathbf{d})$, where $m_{d_i}(\mathbf{d})$ represents the units of work that can be removed from queue $d_i$ in one time slot under priority decision $\mathbf{d}$. $\mathbf{m}(\mathbf{d})$ is akin to $\mathbf{p}(\mathbf{d})$ in our context. However, the generalized switch model in \cite{DiW06} implies properties on the service rate vectors $\mathbf{m}(\mathbf{d})$. For example, it implies that for all $\mathbf{d} = (d_1, d_2, \cdots, d_n)$, we have $m_{d_1}(\mathbf{d}) \geq m_{d_1}(\mathbf{d}^\prime)$ for all $\mathbf{d}^\prime$, and $m_{d_2}(\mathbf{d}) \geq m_{d_2}(\mathbf{d}^\prime)$ for all $\mathbf{d}^\prime$ satisfying $m_{d_1}(\mathbf{d}) = m_{d_1}(\mathbf{d}^\prime)$, etc. These implicit requirements do not necessarily hold in systems which satisfy the conditions in Theorem \ref{thm_sufficient_condition_for_LDF_optimality}. 

If the system has only two users, then clearly subset payoff equivalence is satisfied since the two expected payoff vectors are always on a line. Therefore, we get the following corollary. 
\begin{corollary}
\label{corollary_two_user_optimality}
If the system has two users and satisfies monotonicity in payoffs, then $\mathbf{w}$-LDF policies are feasibility optimal. 
\end{corollary}

Note that in a two-user scenario, the property of monotonicity in payoffs simply means a user gets higher payoff under the higher priority than its payoff under the lower priority. 
In Section \ref{subsection_multiple_classes_of_users} we will consider systems serving two classes of exchangeable users and use this corollary to show the optimality of LDF-like policies. 

\subsection{Efficiency Ratio Analysis}

When the conditions in Theorem \ref{thm_sufficient_condition_for_LDF_optimality} do not hold, one can still study the efficiency ratio, see e.g., \cite{JLS07}, to evaluate the performance of $\mathbf{w}$-LDF policies. 
\begin{definition}
The {\bf efficiency ratio} of the $\mathbf{w}$-LDF policy is defined as
$$
\gamma_{\mathbf{w}\text{-LDF}} = \sup\{\gamma|\gamma \feasibilityRegion \subseteq \feasibilityRegion_{\mathbf{w}\text{-LDF}} \}. 
$$
\end{definition}
Clearly $\gamma_{\mathbf{w}\text{-LDF}}$ equals to $1$ if and only if the $\mathbf{w}$-LDF policy is feasibility optimal.


If a system does not satisfy subset payoff equivalence, i.e., for some subset of users $S$ the vectors in $P^S$ are not on the same hyperplane, we can characterize the ``heterogeneity'' of these vectors based on the following notion. 

\begin{definition}
Given a subset of users $S \subseteq \fullUserSet$, the {\bf subset payoff ratio} $\sigma_S$ for $S$ is defined as
\begin{align}
\label{align_subset_payoff_ratio}
\sigma_S = \max\limits_{{\boldsymbol{\alpha}^S \succeq \mathbf{0}} \atop {\boldsymbol{\alpha}^S \neq \mathbf{0}}} \frac
{\min\limits_{\mathbf{d} \in D(S)} \langle \boldsymbol{\alpha}^S, \mathbf{p}^S(\mathbf{d}) \rangle}
{\max\limits_{\mathbf{d} \in D(S)} \langle \boldsymbol{\alpha}^S, \mathbf{p}^S(\mathbf{d}) \rangle}. 
\end{align}
\end{definition}
The optimal $\boldsymbol{\alpha}^S$ is such that the projections of the vectors in $P^S$ on $\boldsymbol{\alpha}^S$ are as close to each other as possible. 

Clearly if the vectors in $P^S$ are on the same hyperplane, then $\sigma_S = 1$ and the optimal $\boldsymbol{\alpha}^S$ is the normal vector to the hyperplane. Intuitively, $\sigma_S$ characterizes the degree to which the vectors in $P^S$ deviate from being on the same hyperplane. 

This notion enables us to characterize the efficiency ratio of $\mathbf{w}$-LDF for a given system. 

\begin{theorem}
\label{thm_efficiency_ratio}
If the system satisfies monotonicity in payoffs, the efficiency ratio of the $\mathbf{w}$-LDF policy is such that
$$
\gamma_{\mathbf{w}\text{-LDF}} \geq \min\limits_{S \subseteq \fullUserSet} \sigma_S. 
$$
\end{theorem}

\infocomStart
See the extended version of this paper \cite{EXT2} for the proof. 
\commentEnd\fi
\extendedStart
\noindent See Appendix \ref{appendix_pf_thm_efficiency_ratio} for the proof. 
\commentEnd\fi
Intuitively, the bottleneck of the efficiency ratio is the subset $S$ where $\sigma_S$ is the smallest. 

Note that by picking any $\boldsymbol{\alpha} \succ \mathbf{0}$, we can get lower bounds on $\sigma_S$ for all subsets $S\subseteq \fullUserSet$ by placing its projection $\boldsymbol{\alpha}^S$ into (\ref{align_subset_payoff_ratio}). Thus, any $\boldsymbol{\alpha}\succ \mathbf{0}$ enables us to construct a lower bound on $\gamma_{\mathbf{w}\text{-LDF}}$. A trivial option is $\boldsymbol{\alpha} = \mathbf{1}$, where for each subset $S$ the value of $\langle \mathbf{1}^S, \mathbf{p}^S(\mathbf{d}) \rangle$ represents the sum payoff of users in $S$ under decision $\mathbf{d}$. 
\dissertationStart
In the sequel we will consider specific resource and user models in SRT context and explore other options of $\boldsymbol{\alpha}^S$ to evaluate $\mathbf{w}$-LDF's efficiency. 
\commentEnd\fi

We have shown that the efficiency and optimality of the $\mathbf{w}$-LDF policies is related to $R_\text{IB}$. 
Understanding and analyzing the geometry of $R_\text{IB}$ can in principle enable us to provide feedback to the designers of priority-based resource allocation mechanisms regarding which specific priority decision or set of priority decisions are problematic and bottlenecks for the system so that the designers can focus on improving the resource allocation for these problematic decisions. 
For example, in the conceptual setting shown in Figure~{\ref{fig_eg_for_B_intersect_C}}, the priority decision corresponding to the lower left circle is the ``bottleneck'' of the system and should be targeted to make the dominant of the convex hull as small as possible. 
This is of particular interest for some practical systems where it is possible to get explicit knowledge of $P$ which reflect the underlying priority-based resource allocation, e.g., by collecting data over a long time.

A priority decision is problematic if the associated underlying resource allocation suffers from resource contention, blocking among users/applications, or even deadlocks on compute resources, etc. 
Based on feedback regarding the bottlenecks, the designer could improve the associated resource allocation schemes, e.g., by increasing the processing speed of the certain computing resources, spending more energy, reducing the contention, and/or resolving the blocking/deadlock, and thus, improve the efficiency of the overall system under the $\mathbf{w}$-LDF prioritization policies.  

\section{Examples for $\mathbf{w}$-LDF's Optimality}
Theorem \ref{thm_sufficient_condition_for_LDF_optimality} gives a sufficient condition for $\mathbf{w}$-LDF to be feasibility optimal. 
One example system that satisfies these conditions is the model considered in prior work \cite{HoK12} which, as mentioned in Section \ref{sec_introduction}, can be viewed as a single-resource geometric-workload model. 
In this section we consider more system settings and show how our results provide useful insights in practice. 
\dissertationStart
Theorem \ref{thm_sufficient_condition_for_LDF_optimality} gives us a sufficient condition for $\mathbf{w}$-LDF policies to be feasibility optimal. In this section we consider the soft real-time setting and discuss several examples that satisfy those conditions. 

Recall that in SRT setting users are periodically generating streams of tasks that need to complete before deadline and $\mathbf{p}(\mathbf{d})$ represents the expected number of tasks/sub-tasks completed on time, or the expected quality of task processing results per period under priority decision $\mathbf{d}$. To verify the conditions in Theorem \ref{thm_sufficient_condition_for_LDF_optimality}, we verify the subset payoff equivalence property. 
\commentEnd\fi

\subsection{Exchangeable Expected Payoffs}

We shall start by showing that for systems that are ``symmetric'', $\mathbf{w}$-LDF policies are feasibility optimal. 

\begin{definition}
A subset of users $S$ is said to have {\bf exchangeable expected payoffs} if, for all priority decisions $\mathbf{d} \in D$ and all $i, j\in S$, if we switch the priorities of user $i$ and $j$ and use $\mathbf{d}^\prime$ to represent the resulting new priority decision, then
$$
p_k(\mathbf{d}^\prime) = \left\{ 
   \begin{array}{l l}
     p_k(\mathbf{d}) & \quad \text{if $k \neq i,j$}	\\
     p_j(\mathbf{d}) & \quad \text{if $k = i$}	\\
     p_i(\mathbf{d}) & \quad \text{if $k = j$}. 	\\
   \end{array} \right.
$$
\end{definition}
In other words, exchanging the priorities of two users in $S$ will simply exchange their expected payoffs without impacting that of other users. 
This would be true if the priority-based resource allocation were symmetric for users in $S$ and the users generate tasks with identically distributed or exchangeable workloads. 

If the users in $\fullUserSet$ have exchangeable expected payoffs, we can verify the property of subset payoff equivalence by picking $\boldsymbol{\alpha}^S = \mathbf{1}^S$ for each subset of users $S$.
Therefore, by Theorem \ref{thm_sufficient_condition_for_LDF_optimality} we get the following corollary. 

\begin{corollary}
\label{corollary_exchangeable expected_workloads}
If the set of users $\fullUserSet$ have exchangeable expected payoffs and the system satisfies monotonicity in payoffs, then the $\mathbf{w}$-LDF policies are feasibility optimal. 
\end{corollary}

\noindent See Appendix \ref{appendix_pf_corollary_exchangeable expected_workloads} for the proof. 

\subsection{Multiple Classes of Exchangeable Users and Hierarchical-LDF}
\label{subsection_multiple_classes_of_users}
In this subsection, we first consider a system supporting two classes of exchangeable users. 
Formally, a class of users is {\em exchangeable} if they have exchangeable expected payoffs and the same QoS requirement. 
The users in different classes may have distinct payoffs and QoS requirements. 
In some contexts it is of practical interest to first prioritize the classes and then prioritize users in each class, respectively. We refer to such schemes as using {\em class-based hierarchical prioritization}. 

In practice, depending on whether the priorities of classes can change dynamically, there are two types of class-based hierarchical prioritization: Type $1$ where the class priorities are fixed, and Type $2$ where one can dynamically prioritize classes of users, and then users within each class. 

The first type of hierarchical prioritization might correspond to a setting where the users/applications are separated into human-interactive/high-QoS and background-processing/low-QoS categories \cite{PCO}, and it is always desirable to first process high-QoS users. In this setting, the problem is reduced to a collection of independent user prioritization problems similar to the one considered in this paper. By Corollary \ref{corollary_exchangeable expected_workloads}, $\mathbf{w}$-LDF is feasibility optimal to prioritize users in each class. 

The second type of dynamic hierarchical prioritization might be of interest in systems where switching between processing different user classes involves overheads, and/or where it is inefficient to mix the processing of different user classes, probably because of resource contention or deadlocks. 

In this setting, we propose a class-based {\em hierarchical-LDF} policy that in each period works in two steps by (1) prioritizing classes by LDF based on the aggregate deficits, i.e., the sum of deficits for users in the same class, and (2) prioritizing users in each class according to LDF based on individual users' deficits. 
The framework of hierarchical-LDF is exhibited in Figure{~\ref{fig_hierarchical_LDF}}. 
Note that here LDF can be replaced by $\mathbf{w}$-LDF for any $\mathbf{w} \succ \mathbf{0}$ and the following result would hold. 

\begin{figure}[htp]
  \centering
  \includegraphics[width=0.55\textwidth]{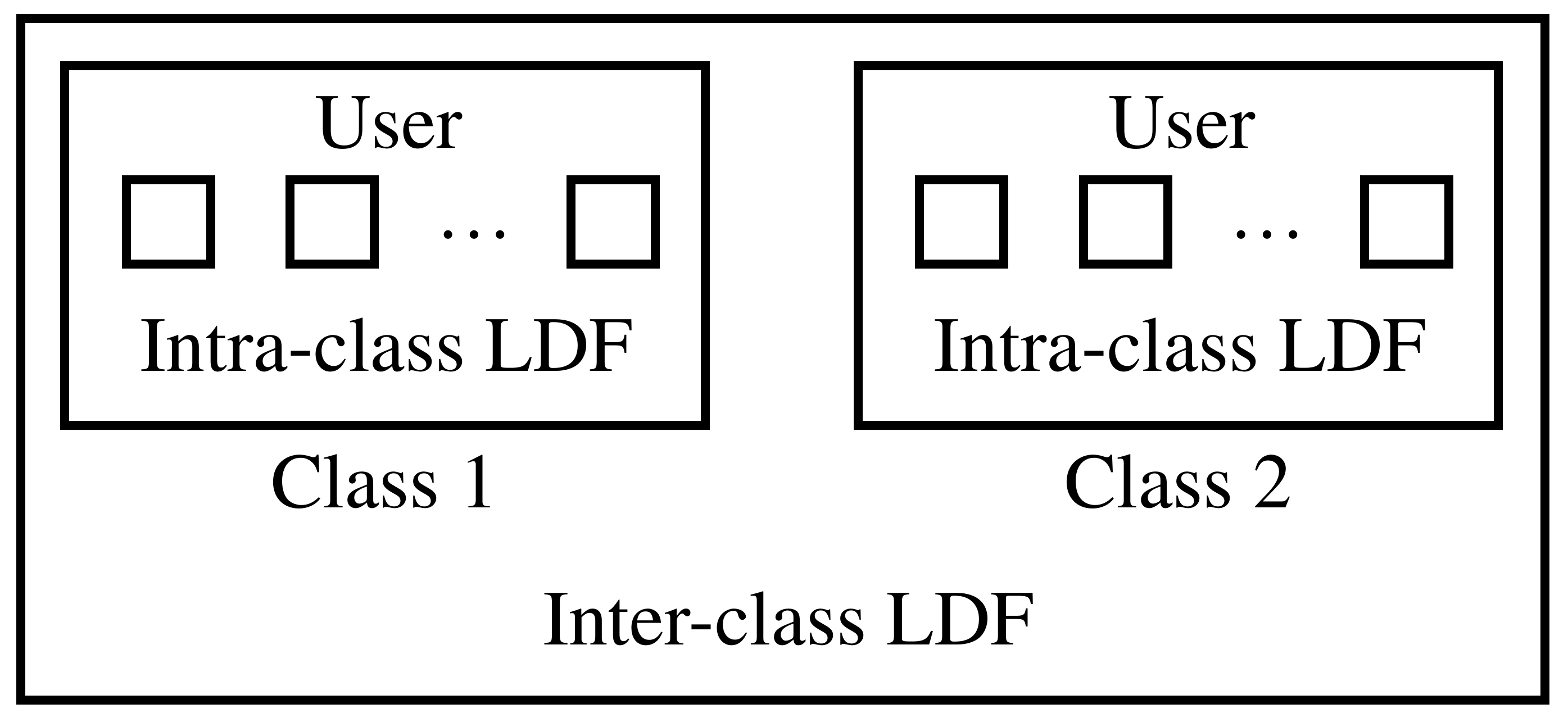}
  \caption{The framework for class-based hierarchical-LDF policy. }
  \label{fig_hierarchical_LDF}
\end{figure} 

\begin{theorem}
\label{thm_optimality_hierarchical_LDF}
In a system with two classes of exchangeable users, if the property of monotonicity in payoffs is satisfied, the hierarchical-LDF policy is feasibility optimal among all possible class-based hierarchical prioritization policies. 
\end{theorem}

The proof follows directly from Corollary \ref{corollary_two_user_optimality} and \ref{corollary_exchangeable expected_workloads}. By Corollary \ref{corollary_two_user_optimality} we know the class-based LDF policy is optimal to set priorities amongst the two classes and by Corollary \ref{corollary_exchangeable expected_workloads} we know the LDF-based user prioritization is also optimal for the exchangeable users in each class. 

More generally, for systems serving multiple (more than two) classes of exchangeable users, 
one can view each class as a ``super user'', and define the aggregate payoff and QoS requirement for a super user to be the sum of payoffs and QoS requirements for users in that class, respectively. 
Then the dynamic prioritization of super users can be viewed as the problem considered in this paper. Therefore, by Theorem \ref{thm_sufficient_condition_for_LDF_optimality}, if the system with the super users' expected aggregate payoffs satisfies monotonicity in payoffs and subset payoff equivalence, the LDF policy is a feasibility optimal choice for prioritizing super users and thus, the hierarchical-LDF policy is feasibility optimal among all class-based hierarchical prioritization policies. 
Indeed, all the results we have introduced, e.g., Theorem \ref{theorem_MW_feasibility_optimal}-\ref{thm_efficiency_ratio}, still hold for the prioritization of these super users. 

\dissertationStart
\subsection{Single-Resource Memoryless-Workload (SRMW) Model}

In the exchangeable expected payoff scenario, we verified subset payoff equivalence by picking $\mathbf{1}^S$ as the normal vector of the hyperplane for each subset of users $S\subseteq \fullUserSet$, i.e., the sum service rate $\langle \mathbf{1}^S, \mathbf{p}^S(\mathbf{d}) \rangle$ is a constant for decisions $\mathbf{d} \in D(S)$. In this subsection we discuss an example that still satisfies subset payoff equivalence but where the sum service rate is not necessarily a constant. 

We consider a special case of the model described in \ref{subsection_SRT_model}. Suppose the system has only one resource, i.e., only one task is processed at a time. Each user generates exactly one task at each period with release time and deadline being the beginning and end of the period, respectively. Again tasks not completed by the deadline are dropped.
We let $\mathbf{p}(\mathbf{d})$ be the expected number of timely completed tasks per period under priority decision $\mathbf{d}$ and suppose each user $i$ requires $q_i$ as the long-term average number of completed tasks per period. 
In this setting the workload of a task refers to the required core time to fully complete the task. Note that a task's workload can be large in which case it may not complete on time. 
Suppose in each period, the workload of the task from user $i$ is a memoryless random variable with parameter $c_i$. (By memoryless random variable we refer to the exponential random variable with $c_i$ being the rate parameter in the continuous-time scenario and the geometric random variable with $c_i$ being the success probability in the discrete-time scenario. ) We assume the workloads are independent across users and periods. At each period given priority decision $\mathbf{d}$, the system processes the tasks sequentially from highest to lowest priority. Such a model is called {\em single-resource memoryless-workload model with parameters $\{c_1, c_2, \cdots, c_n\}$}, or SRMW for short. 

As we have argued in Section \ref{sec_introduction}, the model in prior work \cite{HoK12} can be viewed as a single-resource geometric-workload model, and thus is the discrete version of the SRMW model. In \cite{HoK12} they show that LDF is feasibility optimal, and here in SRMW model we generalize the optimality to the class of $\mathbf{w}$-LDF policies

\begin{corollary}
\label{corollary_SCMW}
In the SRMW model, the class of $\mathbf{w}$-LDF policies are feasibility optimal. 
\end{corollary}

See the appendix for the proof. 
\commentEnd\fi

\dissertationStart
This subsection should also go to dissertation, probably with the following paragraph. 

This SRMW model is a generalized version of the system model in \cite{HoK12}. The work in \cite{HoK12} studies a discrete-time wireless model where $n$ users are competing to transmit a packet to the server at each period of length $\delta$ time slots. The deadlines of the packets are the end of the period. If user $i$ is scheduled at a time slot, it successfully transmits the packet with probability $c_i$ and will re-transmit if it fails until the period ends. Therefore, the workload, i.e., the time to successfully transmit a packet from user $i$ is a geometric random variable with success parameter $c_i$. In \cite{HoK12} they show that LDF is feasibility optimal.
\commentEnd\fi

\section{Some Practical Issues}
In practice, besides meeting minimum payoff requirements, users may be willing to pay for additional payoffs, e.g., better video quality in the video conferencing setting, albeit at possibly different prices. 
Given the requirements $\reqvec$ and the achieved average payoffs $\mathbf{p} = (p_1, p_2, \cdots, p_n)$, we call $p_i - \reqscalar_i$ the {\em excess payoff} for each user $i$. While using $\mathbf{w}$-LDF policies to fulfill users' payoff requirements, we also want to manage the allocation of excess payoffs across users, perhaps with the aim of maximizing the benefits to the system or users. 

However, the non-negative definition of deficit (\ref{align_deficit_in_general_payoff}) makes it hard to track excess payoffs. For example, consider a model with $2$ users and suppose the payoff is always $1$ for the high priority user and $0$ for the low priority user. 
Suppose the payoff requirement vector is $\reqvec = (0.1, 0.5)$. Since $1 > 0.1 + 0.5$, we know $\reqvec$ is feasible and the system can deliver $0.4$ excess payoff. 
Suppose we use the LDF policy, starting from $\mathbf{X}(0) = (0, 0)$ it is easy to verify\footnote{Since the payoffs are deterministic, we can verify this by evaluating the deficits for the first few periods and we will observe that the process $\{\mathbf{X}(t)\}_{t\geq 1}$ evolves in a periodic pattern. } that the system will switch giving high priority to these two users, and thus the achieved average payoff vector is $\mathbf{p} = (0.5, 0.5)$. Clearly User $1$ gets $0.4$ excess payoff while User $2$ gets nothing. 
This happens because $X_1(t)$ and $X_2(t)$ are frequently forced to $0$ from different negative values, which causes the ``unfairness'' between these two users. 

To solve this problem, we modify the deficit definition for each user $i$ and period $t+1$ as follows, 
\begin{align}
\label{align_deficit_possibly_negative}
X_i^\prime(t+1) = X_i^\prime(t) + \reqscalar_i - V_i(\mathbf{d}(t+1)), 
\end{align}
i.e., we allow $X_i^\prime(t)$ to be negative. 

Now for the simple example above, if we adopt LDF but based on the possibly negative deficits $\mathbf{X}^\prime(t) = (X_1^\prime(t), X_2^\prime(t), \cdots, X_n^\prime(t))$, we can get achieved average payoff vector $\mathbf{p} = (0.3, 0.7)$. We observe that the two users equally split the excess payoff. 

Intuitively, for each user $i$ the modified deficit $X_i^\prime(t)$ changes roughly linearly as $t$ increases with the slope being $\reqscalar_i - p_i$. Since $\mathbf{w}$-LDF policy aims to balance weighted deficit $w_iX_i^\prime(t)$, we know $w_i(p_i - \reqscalar_i)$ is roughly the same for all users.
We will verify this observation in the simulation section and based on this we can manage the excess payoffs across users by picking the appropriate weight vector $\mathbf{w}$. 

\dissertationStart
In this section, we will cover some additional practical issues associated with $\mathbf{w}$-LDF policies.  

Besides meeting minimum long-term payoff requirements, users in practice may also be willing to pay for additional payoffs, i.e., better QoS, albeit at possibly different prices. 
For example, in the video conferencing setting, users may be willing to pay for better video quality if possible. 
For users that require average payoffs $\reqvec$ but actually achieve average payoffs $\mathbf{p} = (p_1, p_2, \cdots, p_n)$, we call $p_i - \reqscalar_i$ the {\em excess payoff} for each user $i$. While using $\mathbf{w}$-LDF policies to fulfill users' payoff requirements, we also want to manage the allocation of excess payoffs across users, perhaps with the aim of maximizing the relative benefits to the system or users. 

However, the non-negative definition of deficit (\ref{align_deficit_in_general_payoff}) makes it hard to track excess payoffs. For example, consider a model with $2$ users and suppose the payoff is always $1$ for the high priority user and always $0$ for the low priority user. 
Suppose the payoff requirement vector is $\reqvec = (0.1, 0.5)$. Since $1 > 0.1 + 0.5$, we know this payoff requirement is feasible and the system can deliver $0.4$ excess payoff. 
Suppose we use the LDF policy, starting from $\mathbf{X}(0) = (0, 0)$ it is easy to verify\footnote{Since the payoffs are deterministic, we can verify this by evaluating the deficits for the first few periods and we will observe that the process $\{\mathbf{X}(t)\}_{t\geq 1}$ evolves in a periodic pattern. } that the system will switch giving high priority to these two users from period to period, giving the achieved average payoff vector $\mathbf{p} = (0.5, 0.5)$. Clearly User $1$ gets $0.4$ excess payoff while User $2$ gets nothing. 
This happens because $X_1(t)$ and $X_2(t)$ are frequently forced to $0$ from different negative values, which causes the ``unfairness'' between these two users. 

To solve this problem, we modify the deficit definition for each user $i$ and period $t+1$ as follows, 
\begin{align}
\label{align_deficit_possibly_negative}
X_i^\prime(t+1) = X_i^\prime(t) + \reqscalar_i - V_i(\mathbf{d}(t+1)), 
\end{align}
i.e., we allow $X_i^\prime(t)$ to be negative. 

Now for the simple example above, if we adopt LDF but based on the possibly negative deficits $\mathbf{X}^\prime(t) = (X_1^\prime(t), X_2^\prime(t), \cdots, X_n^\prime(t))$, we can get achieved average payoff vector $\mathbf{p} = (0.3, 0.7)$. We observe that the two users equally split the excess payoffs. 

Intuitively, for each user $i$ the modified deficit $X_i^\prime(t)$ changes roughly linearly as $t$ increases with the slope being $\reqscalar_i - p_i$. Since $\mathbf{w}$-LDF policy aims to balance weighted deficit $w_iX_i^\prime(t)$, we know $w_i(p_i - \reqscalar_i)$ is roughly the same for all users, unless some $p_i$ cannot be increased/decreased any more which means user $i$ is already always assigned the highest/lowest priority. We will verify this observation in the simulation section and based on this we can manage the excess payoffs for all users $i$ by picking the appropriate weight vector $\mathbf{w}$. 
\commentEnd\fi

\add{
Note that for completeness we will need to modify the feasibility definition since the process $\{\mathbf{X}^\prime(t)\}_{t\geq 1}$ is no longer positive recurrent as it may keep decreasing or increasing. 
\infocomStart
Refer to the extended version of this paper \cite{EXT2} for details. 
\commentEnd\fi
\extendedStart
Now we call a payoff requirement vector $\reqvec$ feasible if, under some user prioritization policy, for each user $i$ the time-averaged payoff per period is at least $\reqscalar_i$. Formally, recall that $V_i(\mathbf{d}(t))$ is the random payoff for user $i$ in period $t$. As the payoff requirement, each user $i$ requires that
$$
\liminf_{\tau \rightarrow \infty} \frac{1}{\tau}\sum\limits_{t = 1}^{\tau} V_i(\mathbf{d}(t)) \geq q_i, \text{with probability 1}.
$$
Note that this definition and Definition \ref{defn_feasibility_pr} are just two ways to define the feasibility. With the theorem in \cite{Bla56} we can show that for any user prioritization policy, the sets of feasible QoS requirements under these two different feasibility definitions differ by at most a boundary and thus are equivalent for practical purposes. Therefore, all the results we discuss in this paper hold under both feasibility definitions. 
\commentEnd\fi
}

\section{Simulations}
\label{sec_simulations}
In this section we explore via simulation the impact of weights of $\mathbf{w}$-LDF policies. 

Consider an illustrative system with single computing resource serving $3$ soft real-time users. In each period of length $\delta = 10$, each user generates one task that need to complete by end of the period. We let the non-negative workload, i.e., task service time, distributions for three users be Gamma$(12, 0.5)$, Gamma$(4, 1)$ and Gamma$(10, 0.1)$, respectively. We pick these workload distributions to make them general and heterogeneous. 
In each period, the payoff for user $i$ is $1$ if user $i$'s task completes and is $0$ otherwise. Accordingly, user $i$'s QoS requirement $\reqscalar_i$ represents the long-term task completion ratio. 

We start with initial deficit $\mathbf{X}^\prime(0) = (0, 0, 0)$. In each period, we independently generate task workloads for users and simulate the $\mathbf{w}$-LDF policy based on $\mathbf{X}^\prime(t)$ to pick a priority decision. The single resource sequentially processes users' tasks from highest to lowest priority. 
Tasks not completed on time are dropped. 
All simulations are run for $30000$ periods.  
A requirement vector $\reqvec$ is feasible if it is dominated by the achieved task completion ratio vector $\mathbf{p}$ over the $30000$ periods. The vectors $\reqvec$ and $\mathbf{w}$ are specified in various settings in the sequel. 

Note that in this setting monotonicity in payoffs is satisfied while subset payoff equivalence is not. 

\subsection{Impact of Weights on Long-Term Completion Ratios}
\label{subsection_effect_weights_long_term}

In Table \ref{tab_q_and_q_prime} we consider a requirement vector $\reqvec$ that is feasible under the $\mathbf{w}$-LDF policies and display the achieved 
$\mathbf{p}$ under two different weight vectors $\mathbf{w}$. For each weight vector $\mathbf{w}$, we verify that $w_i(p_i - \reqscalar_i)$ is the same for all three users. 
Contrasting the two lines in Table \ref{tab_q_and_q_prime}, we can see that for a system which can deliver more than required, changing the weight vector reallocates the excess payoffs and gives more excess payoffs to users with smaller weights. 

\begin{table}[h]
\normalsize
	\tbl{Achieved completion ratio vectors under two weight vectors. \label{tab_q_and_q_prime}}{
    \centering
	\begin{tabular}{|c|c|c|c|}
	 \hline
	 $\mathbf{q}$ & $\mathbf{w}$ & Achieved $\mathbf{p}$ & $w_i(p_i - \reqscalar_i)$	\\
	 \hline
	 \hline
	 \multirow{2}{*}{$0.8, 0.6, 0.4$} & $(1,1,1)$ & $0.85, 0.65, 0.45$ & $0.05$	\\
	 \cline{2-4}
	  & $(10,1,1)$ & $0.809, 0.69, 0.49$ & $0.09$	\\
	 \hline
	\end{tabular}
	}
\end{table}

\dissertationStart
\subsection{Impact of Weights on Long-Term Achieved Completion Ratios}
\label{subsection_effect_weights_long_term}
\myComments{This needs further improvement. }
We denote by $\mathbf{q}^\prime = (q^\prime_1, q^\prime_2, \cdots, q^\prime_n)$ the actually achieved long-term time-averaged task completion ratio vector. 
In scenarios where each user has a positive time fraction of being assigned the highest priority under $\mathbf{w}$-LDF policy\footnote{In some scenarios, this may not be true. For example, in a two-user scenario, if the requirement of user $1$ is much larger than user $2$ the $\mathbf{w}$-LDF policy may end up always assigning user $1$ with the higher priority. But we claim most real-time applications in practice requires a reasonable QoS requirement, making the requirements similar to each other, and therefore, these scenarios are not of great interest. }, we observe that $w_i(q_i - q^\prime_i)$ is the same for each user $i$, where $q_i$ and $q^\prime_i$ represent the required and achieved completion ratio for user $i$, respectively. 
This is because deficit $x_i(t)$ roughly\footnote{Because of the randomness, it is not exactly linear change. } changes linearly over periods with slope being $q_i - q^\prime_i$, as shown in Figure~\ref{figure_example_deficit}, and $\mathbf{w}$-LDF policy aims to balance weighted deficit $w_ix_i(t)$. 
This is true for both feasible and infeasible requirements. 

\begin{figure}[htp]
  \centering
  \includegraphics[width=0.45\textwidth]{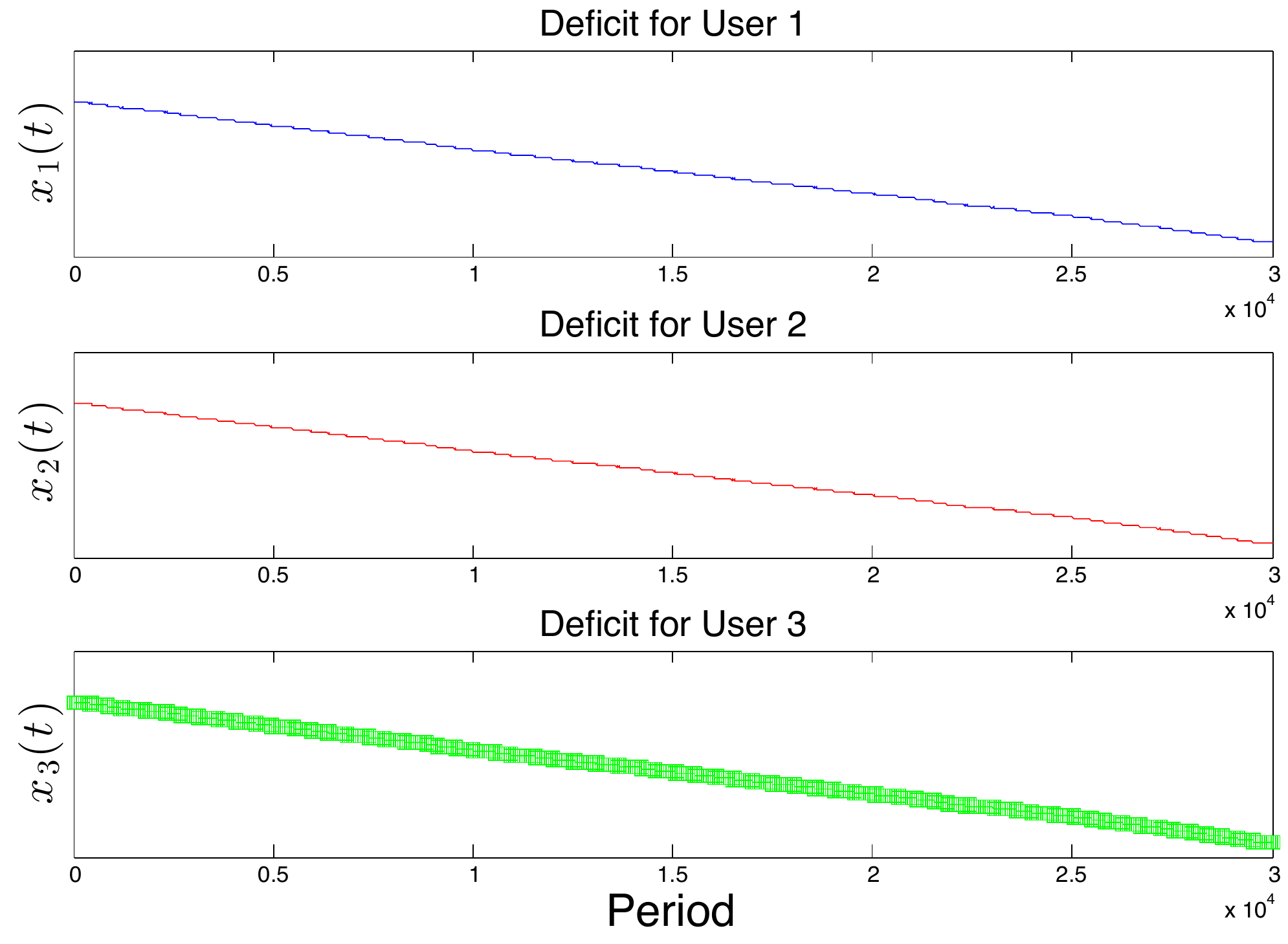}
  \caption{An example of the deficits of users in a three-user scenario with a feasible vector of requirements. }
  \label{figure_example_deficit}
\end{figure}

For example one can verify this statement in Table \ref{tab_q_and_q_prime} where we display the achieved completion ratio vectors for a feasible and an infeasible vector of requirements under two different weight vectors. As can be seen, increasing the weight of user $i$ drives the achieved completion ratio $q^\prime_i$ to the 
required $q_i$ even if the requirement vector $\mathbf{q}$ is infeasible, but with the risk of penalizing other users. Conversely, a user with a smaller weight will likely experience a larger deficit or additional completions, depending on the feasibility of the requirements. 



\begin{table}[h]
\normalsize
    \centering
	\begin{tabular}{|c|c|c|c|}
	 \hline
	 $\mathbf{q}$ & $\mathbf{w}$ & Achieved $\mathbf{q}^\prime$ & $w_i(q_i - q^\prime_i)$	\\
	 \hline
	 \multirow{4}{*}{} &
	 \multirow{2}{*}{$(1,1,1)$} & \multirow{2}{*}{$0.86, 0.84, 0.82$} & \multirow{2}{*}{$-0.02$}	\\
	 Feasible: &  &  &	\\
	 \cline{2-4}
	 $0.84, 0.82, 0.80$ & \multirow{2}{*}{$(10,1,1)$} & \multirow{2}{*}{$0.843, 0.85, 0.83$} & \multirow{2}{*}{$-0.03$}	\\
	  &  &  &	\\
	 \hline
	 \multirow{4}{*}{} &
	 \multirow{2}{*}{$(1,1,1)$} & \multirow{2}{*}{$0.86, 0.84, 0.82$} & \multirow{2}{*}{$0.04$}	\\
	 Infeasible: &  &  &	\\
	 \cline{2-4}
	 $0.90, 0.88, 0.86$ & \multirow{2}{*}{$(10,1,1)$} & \multirow{2}{*}{$0.893, 0.81, 0.79$} & \multirow{2}{*}{$0.07$}	\\
	  &  &  &	\\
	 \hline
	\end{tabular}
    \caption{Actually achieved completion ratio vectors and the value of $w_i(q_i - q^\prime_i)$ for two QoS requirement vectors under two weight vectors. }
    \label{tab_q_and_q_prime}
\end{table}

Thus, in practice when it is not clear whether a vector of QoS requirements is feasible, one can assign different weights to users based on how the users behave to the gaps between QoS requirements and the achieved completion ratios. 

\commentEnd\fi

\subsection{Characterization of Clustering of Failures and Impact of Weights}
\label{subsection_characterization_of_failure_clustering}
If a user's task is not completed in a period, we call it a failure event. The requirement vector $\reqvec$ focuses on long-term task completion ratio, but it would likely be undesirable for a user to experience consecutive or clustered failure events. 
Figure~\ref{figure_clustering_of_failures} gives an example of failure events. 
In this subsection we consider the same $\mathbf{q} = (0.8, 0.6, 0.4)$ used above and explore the clustering of failures under two $\mathbf{w}$-LDF policies. 

\begin{figure}[htp]
  \centering
  \includegraphics[width=0.6\textwidth]{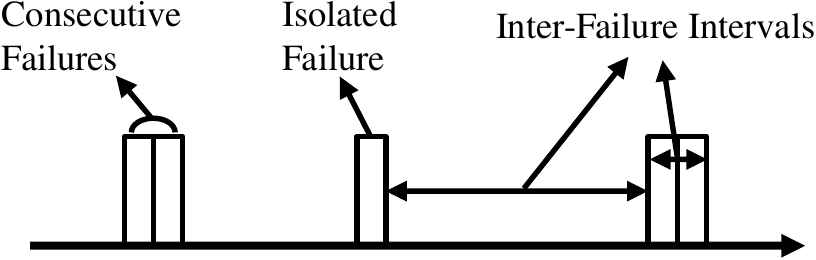}
  \caption{Characteristics of clustering of failures. }
  \label{figure_clustering_of_failures}
\end{figure}

We consider Inter-Failure Intervals (IFIs) between typical failures. 
IFI is supported on the set $\{1, 2, 3, \cdots\}$. To quantitatively evaluate the clustering of the failures, we focus on the standard deviation (SD) of the IFIs for each user. 
One extreme case is that failures happen strictly periodically and therefore, the SD is $0$. 
Intuitively, a user with a smaller IFI SD implies that the user experiences less clustered failures. 

Next we introduce an evaluation benchmark. 
For each user $i$, we know $1 - p_i$ represents the time-averaged failure ratio. 
If the failure happens in each period independently with probability $1 - p_i$, the IFI can be modeled by a geometric random variable supported on the set $\{1, 2, 3, \cdots\}$ with the parameter being $1 - p_i$. We use the SD of such a geometric random variable as a benchmark.  

Under some $\mathbf{w}$-LDF policy, we define {\em SD ratio} of user $i$ to be the ratio of user $i$'s IFI SD to the SD of the geometric random variable with parameter $1 - p_i$. 
Table \ref{tab_sd_ifi} shows the SD ratios of three users under two different weight vectors $\mathbf{w}$. 
Under $\mathbf{w} = (1,1,1)$, the ratios are less than $1$, indicating that the failures under the LDF policy are less clustered compared to the scheme where failure event happens i.i.d. in each period. 
The last two columns in Table \ref{tab_sd_ifi} indicates that increasing the weight of user $i$ reduces the degree of failure clustering for user $i$ but at the price of other users' more clustered failures. Thus, the users' sensitivities to clustered failures is another factor to consider when one assigns weights to users. 

\begin{table}[h]
\normalsize
    \tbl{Characterization of clustering of failures. \label{tab_sd_ifi}}{
    \centering
    \begin{tabular}{| c || c | c |}
    \hline
    ~ & {\bf SD ratio under} & {\bf SD ratio under}\\
    ~ & $\mathbf{w} = (1, 1, 1)$ & $\mathbf{w} = (10, 1, 1)$	\\
    \hline
    \hline
    {\bf User 1} & 88\% & 39\%	\\
    \hline
    {\bf User 2} & 77\% & 97\%	\\
    \hline
    {\bf User 3} & 92\% & 107\%	\\
    \hline
    \end{tabular}
    }
\end{table}

\dissertationStart
\subsection{Characterization of Clustering of Failures and Impact of Weights}
\label{subsection_characterization_of_failure_clustering}
If a user's task is not completed in a period, we call it a failure event. The QoS requirement focuses on time-averaged task completion ratio, but it would likely be undesirable for a user to experience consecutive or clustered failure events. Here we explore via simulation the clustering of failures.
We first consider LDF policy, i.e., $\mathbf{w} = (1,1,1)$, and the feasible QoS requirement vector $\mathbf{q} = (0.84, 0.82, 0.80)$ used in \ref{subsection_effect_weights_long_term}. Figure~ \ref{figure_clustering_of_failures} gives an example of failure events. 

\begin{figure}[htp]
  \centering
  \includegraphics[width=0.45\textwidth]{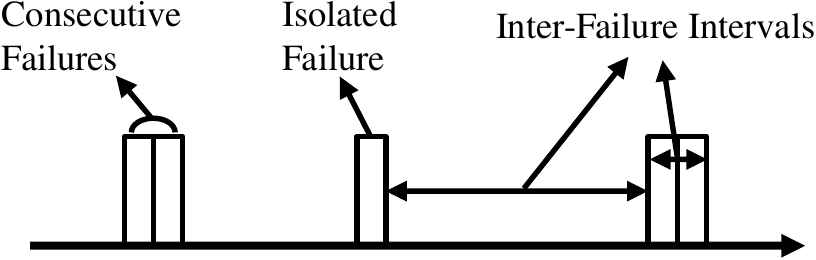}
  \caption{Characteristics of clustering of failures. }
  \label{figure_clustering_of_failures}
\end{figure}

We first consider sequences of consecutive failures (including a single failure event) and display the distributions of consecutive failures for our three users in Figure~{\ref{figure_consecutive_failures}}. Clearly isolated failure instances are the most likely (i.e., with probability $90\%$, $90\%$, $86\%$ for three users, respectively) but there is a chance (albeit small) of seeing $4$ or $5$ consecutive failures. 

\begin{figure}[htp]
  \centering
  \includegraphics[width=0.4\textwidth]{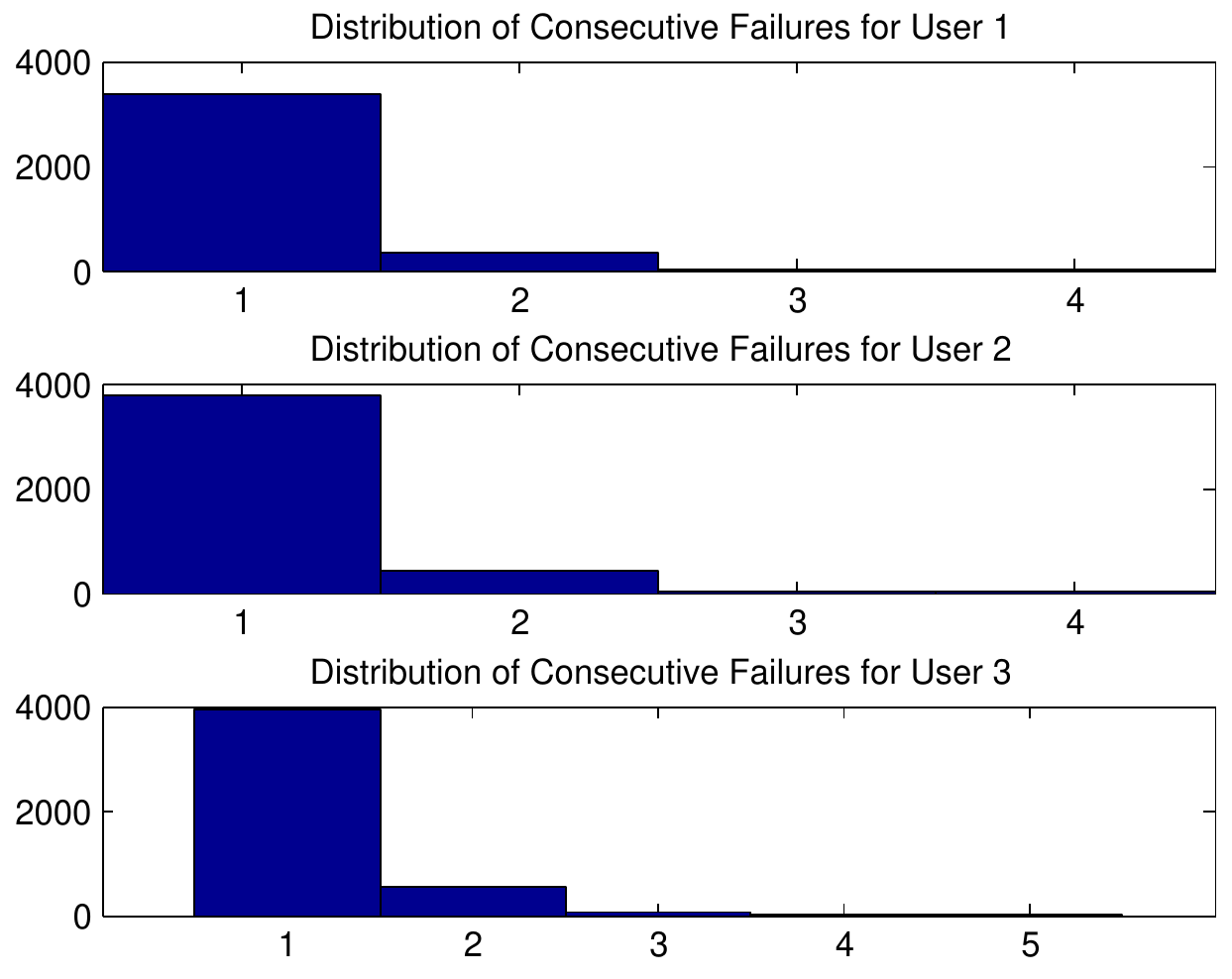}
  \caption{The distribution of consecutive failures. }
  \label{figure_consecutive_failures}
\end{figure}

Next we consider Inter-Failure Intervals (IFIs) between typical failures. 
IFI is supported on the set $\{1, 2, 3, \cdots\}$. To quantitatively evaluate the clustering of the failures, we focus on the standard deviation of the IFIs. 
One extreme case is that failures happen periodically and therefore, the standard deviation is $0$. 
We consider a benchmark called i.i.d. randomly distributed failures scheme (IID-F) where failure happens in each period independently with the same probability. 
In this benchmark IFI is a geometrically distributed random variable. 
To make it a fair comparison, we set the parameter of the benchmark to be the actual time-averaged failure probability under LDF. 
We further define {\em relative degree of failure clustering} for a user $i$ to be the ratio of the standard deviation of user $i$'s IFIs under LDF to that in IID-F.
The standard deviations of IFIs under LDF and IID-F and their ratios are shown in Table \ref{tab_sd_ifi}. 
The results exhibit that under LDF policy the failures are less clustered compared to the IID-F benchmark and the three users have similar relative degree of failure clustering. 

\begin{table}[h]
\normalsize
    \centering
    \begin{tabular}{| c | c | c | c |}
    \hline
    ~ & {\bf LDF Policy} & {\bf IID-F} & {\bf LDF/IID-F}\\
    \hline
    {\bf User 1} & 5.83 & 6.77 & 86\%	\\
    \hline
    {\bf User 2} & 4.90 & 5.84 & 84\%	\\
    \hline
    {\bf User 3} & 4.40 & 5.12 & 86\%	\\
    \hline
    \end{tabular}
    \caption{Standard deviations of inter failure intervals with $\mathbf{w} = (1,1,1)$. }
    \label{tab_sd_ifi}
\end{table}

Changing the weights also impacts the clustering of failures. 

We consider the same requirement $\mathbf{q} = (0.84, 0.82, 0.80)$ but with a different weight vector $\mathbf{w} = (10, 1, 1)$. We conduct the same analysis for IFIs and exhibit the results in Table \ref{tab_sd_ifi_different_weights}. 
The standard deviations in IID-F are different from those in Table \ref{tab_sd_ifi} because of the different achieved completion ratios driven by the weight vectors. 
The ratios in the last columns of Table \ref{tab_sd_ifi} and \ref{tab_sd_ifi_different_weights} indicate that increasing the weight of user $i$ induces a smaller relative degree of failure clustering for user $i$ but at the price of other users' more clustered failures. Thus, the users' sensitivities to clustered failures is another factor to consider when one assigns weights to users. 

\begin{table}[h]
\normalsize
    \centering
    \begin{tabular}{| c | c | c | c |}
    \hline
    ~ & {\bf $\mathbf{w}$-LDF Policy} & {\bf IID-F} & {\bf $\mathbf{w}$-LDF/IID-F}\\
    \hline
    {\bf User 1} & 4.43 & 5.85 & 76\%	\\
    \hline
    {\bf User 2} & 7.27 & 6.21 & 117\%	\\
    \hline
    {\bf User 3} & 5.93 & 5.41 & 110\%	\\
    \hline
    \end{tabular}
    \caption{Standard deviations of inter failure intervals with $\mathbf{w} = (10,1,1)$. }
    \label{tab_sd_ifi_different_weights}
\end{table}
\commentEnd\fi

\section{Conclusion}
Resource allocation in complex systems supporting real-time users with general QoS requirements can be relatively ``easy''. One can in principle design the system to allow priority-based resource allocation and adopt simple $\mathbf{w}$-LDF policies to dynamically prioritize users/applications. 
Our theory provides guidance towards understanding the suboptimality and even optimality of such solutions and how to improve the system design. 
For future work, it would be interesting to explore the management of real-time users across systems and/or sharing with non real-time traffic. 

%

\huaweiStart
\section*{Acknowledgment}
This research was supported by Huawei Technologies Co. Ltd. 
The authors would like to thank Alan Gatherer, Zheng Lu, Haishan Zhu and Mattan Erez for their comments and feedbacks on this work. 
\commentEnd\fi

\bibliography{diss_myown}{}


\begin{thebibliography}{00}


\ifx \showCODEN    \undefined \def \showCODEN     #1{\unskip}     \fi
\ifx \showDOI      \undefined \def \showDOI       #1{{\tt DOI:}\penalty0{#1}\ }
  \fi
\ifx \showISBNx    \undefined \def \showISBNx     #1{\unskip}     \fi
\ifx \showISBNxiii \undefined \def \showISBNxiii  #1{\unskip}     \fi
\ifx \showISSN     \undefined \def \showISSN      #1{\unskip}     \fi
\ifx \showLCCN     \undefined \def \showLCCN      #1{\unskip}     \fi
\ifx \shownote     \undefined \def \shownote      #1{#1}          \fi
\ifx \showarticletitle \undefined \def \showarticletitle #1{#1}   \fi
\ifx \showURL      \undefined \def \showURL       #1{#1}          \fi

\bibitem[\protect\citeauthoryear{Blackwell}{Blackwell}{1956}]%
        {Bla56}
{David Blackwell}. 1956.
\newblock \showarticletitle{{An Analog of the minimax theorem for vector
  payoffs}}.
\newblock {\it Pacific J. Math.} {6}, 1 (November 1956), 1--8.
\newblock


\bibitem[\protect\citeauthoryear{Boyd and Vandenberghe}{Boyd and
  Vandenberghe}{2009}]%
        {BoV09}
{Stephen Boyd} {and} {Lieven Vandenberghe}. 2009.
\newblock {\em {Convex Optimization}}.
\newblock Cambridge university press.
\newblock


\bibitem[\protect\citeauthoryear{Conway and Sloane}{Conway and Sloane}{2013}]%
        {CoS13b}
{J.H. Conway} {and} {N.J.A. Sloane}. 2013.
\newblock {\em {Sphere Packings, Lattices and Groups}}.
\newblock Springer.
\newblock


\bibitem[\protect\citeauthoryear{Dai and Prabhakar}{Dai and Prabhakar}{2000}]%
        {DaP00}
{J.G. Dai} {and} {Balaji Prabhakar}. 2000.
\newblock \showarticletitle{{The throughput of data switches with and without
  speedup}}. In {\em Proceedings of INFOCOM 2000}. 556--564.
\newblock


\bibitem[\protect\citeauthoryear{Davis and Burns}{Davis and Burns}{2011}]%
        {DaB11}
{Robert~I. Davis} {and} {Alan Burns}. 2011.
\newblock \showarticletitle{{A Survey of Hard Real-Time Scheduling for
  Multiprocessor Systems}}.
\newblock {\it Comput. Surveys}  {43} (October 2011).
\newblock
Issue 4.


\bibitem[\protect\citeauthoryear{Dimakis and Walrand}{Dimakis and
  Walrand}{2006}]%
        {DiW06}
{Antonis Dimakis} {and} {Jean Walrand}. 2006.
\newblock \showarticletitle{{Sufficient Conditions for Stability of
  Longest-Queue-First Scheduling: Second-Order Properties Using Fluid Limits}}.
\newblock {\em Advances in Applied Probability\/} {38}, 2 (June 2006).
\newblock


\bibitem[\protect\citeauthoryear{Down and Meyn}{Down and Meyn}{1994}]%
        {DoM94}
{D. Down} {and} {S. Meyn}. 1994.
\newblock \showarticletitle{{A Survey of Markovian Methods for Stability of
  Networks}}. In {\em 11th International Conference on Analysis and
  Optimization of Systems}.
\newblock


\bibitem[\protect\citeauthoryear{Down and Meyn}{Down and Meyn}{1997}]%
        {DoM97}
{D. Down} {and} {S.P. Meyn}. 1997.
\newblock \showarticletitle{{Piecewise linear test functions for stability and
  instability of queueing networks}}.
\newblock {\em Queueing Systems\/}  {27} (April 1997), 205--226.
\newblock
Issue 3-4.


\bibitem[\protect\citeauthoryear{Du and de~Veciana}{Du and de~Veciana}{2016}]%
        {DuD16S}
{Yuhuan Du} {and} {Gustavo de Veciana}. 2016.
\newblock \showarticletitle{{Scheduling for Cloud-Based Computing Systems to
  Support Soft Real-Time Applications}}.
\newblock {\em INFOCOM 2016\/} (April 2016).
\newblock


\bibitem[\protect\citeauthoryear{Gatherer}{Gatherer}{2015}]%
        {PCO}
{Alan Gatherer}. 2015.
\newblock Personal communication.   (February 2015).
\newblock


\bibitem[\protect\citeauthoryear{Hou and Kumar}{Hou and Kumar}{2012}]%
        {HoK12}
{I-Hong Hou} {and} {P.~R. Kumar}. 2012.
\newblock \showarticletitle{{Queueing systems with hard delay constraints: a
  framework for real-time communication over unreliable wireless channels}}.
\newblock {\em Queueing Systems\/}  {71} (March 2012), 151--177.
\newblock
Issue 1-2.


\bibitem[\protect\citeauthoryear{Hou and Kumar}{Hou and Kumar}{2013}]%
        {HoK13b}
{I-Hong Hou} {and} {P.~R. Kumar}. 2013.
\newblock {\em {Packets with Deadlines: A Framework for Real-Time Wireless
  Networks}}.
\newblock Morgan \& Claypool Publishers.
\newblock


\bibitem[\protect\citeauthoryear{Hou and Kumar}{Hou and Kumar}{2014}]%
        {HoK14}
{I-Hong Hou} {and} {P.~R. Kumar}. 2014.
\newblock \showarticletitle{{Scheduling Heterogeneous Real-Time Traffic over
  Fading Wireless Channels}}.
\newblock {\em IEEE/ACM Transactions on Networking\/}  {22} (October 2014),
  1631--1644.
\newblock
Issue 5.


\bibitem[\protect\citeauthoryear{Jaramillo and Srikant}{Jaramillo and
  Srikant}{2011}]%
        {JaS11}
{Juan~Jose Jaramillo} {and} {R. Srikant}. 2011.
\newblock \showarticletitle{{Optimal Scheduling for Fair Resource Allocation in
  Ad Hoc Networks With Elastic and Inelastic Traffic}}.
\newblock {\em IEEE Transactions on Networking\/}  {19} (August 2011),
  1125--1136.
\newblock
Issue 4.


\bibitem[\protect\citeauthoryear{Joo, Lin, and Shroff}{Joo
  et~al\mbox{.}}{2007}]%
        {JLS07}
{Changhee Joo}, {Xiaojun Lin}, {and} {Ness~B. Shroff}. 2007.
\newblock \showarticletitle{{Performance Limits of Greedy Maximal Matching in
  Multi-hop Wireless Networks}}. In {\em IEEE Conference on Decision and
  Control}. 1128--1133.
\newblock


\bibitem[\protect\citeauthoryear{Kang, Wang, Jaramillo, and Ying}{Kang
  et~al\mbox{.}}{2013}]%
        {KWJ13}
{Xiaohan Kang}, {Weina Wang}, {Juan~Jose Jaramillo}, {and} {Lei Ying}. 2013.
\newblock \showarticletitle{{On the Performance of Largest-Deficit-First for
  Scheduling Real-Time Traffic in Wireless Networks}}. In {\em Proceedings of
  MobiHoc}. 99--108.
\newblock


\bibitem[\protect\citeauthoryear{Li and Ierapetritou}{Li and
  Ierapetritou}{2008}]%
        {LiI08}
{Zukui Li} {and} {Marianthi Ierapetritou}. 2008.
\newblock \showarticletitle{Process scheduling under uncertainty: Review and
  challenges}.
\newblock {\em Computers and Chemical Engineering\/}  {32} (2008), 715--727.
\newblock
Issue 4-5.


\bibitem[\protect\citeauthoryear{McKeown}{McKeown}{1995}]%
        {MCK95}
{Nicholas McKeown}. 1995.
\newblock \showarticletitle{{Scheduling Algorithms for Input-Queued Cell
  Switches}}.
\newblock  (1995).
\newblock
\newblock
\shownote{Ph.D. dissertation.}


\bibitem[\protect\citeauthoryear{McKeown, Mekkittikul, Anantharam, and
  Walrand}{McKeown et~al\mbox{.}}{1999}]%
        {MMA99}
{Nick McKeown}, {Adisak Mekkittikul}, {Venkat Anantharam}, {and} {Jean
  Walrand}. 1999.
\newblock \showarticletitle{{Achieving 100\% Throughput in an Input-Queued
  Switch}}.
\newblock {\em IEEE Transactions on Communications\/} {47}, 8 (August 1999),
  1260--1267.
\newblock


\bibitem[\protect\citeauthoryear{Meyn and Tweedie}{Meyn and Tweedie}{2008}]%
        {MeT08}
{S.~P. Meyn} {and} {R.~L. Tweedie}. 2008.
\newblock {\em {Markov Chains and Stochastic Stability}}.
\newblock Cambridge University Press.
\newblock


\bibitem[\protect\citeauthoryear{Munir, Lin, Hoque, Nirjon, Stankovic, and
  Whitehouse}{Munir et~al\mbox{.}}{2010}]%
        {MLH10}
{Sirajum Munir}, {Shan Lin}, {Enamul Hoque}, {S.~M.~Shahriar Nirjon}, {J.~A.
  Stankovic}, {and} {K. Whitehouse}. 2010.
\newblock \showarticletitle{{Addressing Burstiness for Reliable Communication
  and Latency Bound Generation in Wireless Sensor Networks}}. In {\em IPSN
  2010}. 303--314.
\newblock


\bibitem[\protect\citeauthoryear{Neely}{Neely}{2009}]%
        {Nee09}
{M.~J. Neely}. 2009.
\newblock \showarticletitle{{Delay Analysis for Max Weight Opportunistic
  Scheduling in Wireless Systems}}.
\newblock {\it IEEE Trans. Automat. Control}  {54} (September 2009),
  2137--2150.
\newblock
Issue 9.


\bibitem[\protect\citeauthoryear{Sha, Abdelzaher, {\r{A}}RZ{\'{E}}N, Cervin,
  Baker, Burns, Buttazzo, Caccamo, Lehoczky, and Mok}{Sha
  et~al\mbox{.}}{2004}]%
        {SAA04}
{Lui Sha}, {Tarek Abdelzaher}, {Karl-Erik {\r{A}}RZ{\'{E}}N}, {Anton Cervin},
  {Theodore Baker}, {Alan Burns}, {Giorgio Buttazzo}, {Marco Caccamo}, {John
  Lehoczky}, {and} {Aloysius~K. Mok}. 2004.
\newblock \showarticletitle{{Real Time Scheduling Theory: A Historical
  Perspective}}.
\newblock {\em Real-Time Systems\/}  {28} (November-December 2004), 101--155.
\newblock
Issue 2-3.


\bibitem[\protect\citeauthoryear{Shakkottai and Srikant}{Shakkottai and
  Srikant}{2002}]%
        {ShS02}
{Sanjay Shakkottai} {and} {R. Srikant}. 2002.
\newblock \showarticletitle{{Scheduling Real-Time Traffic With Deadlines over a
  Wireless Channel}}.
\newblock {\em Wireless Networks\/}  {8} (January 2002), 13--26.
\newblock
Issue 1.


\bibitem[\protect\citeauthoryear{Stolyar}{Stolyar}{2004}]%
        {Sto04}
{Alexander~L. Stolyar}. 2004.
\newblock \showarticletitle{{Maxweight Scheduling in a Generalized Switch:
  State Space Collapse and Workload Minimization in Heavy Traffic}}.
\newblock {\em The Annals of Applied Probability\/} {14}, 1 (February 2004).
\newblock


\bibitem[\protect\citeauthoryear{Tassiulas and Ephremides}{Tassiulas and
  Ephremides}{1992}]%
        {TaE92}
{Leandros Tassiulas} {and} {Anthony Ephremides}. 1992.
\newblock \showarticletitle{{Stability Properties of Constrained Queueing
  Systems and Scheduling Policies for Maximum Throughput in Multihop Radio
  Networks}}.
\newblock {\it IEEE Trans. Automat. Control}  {37} (December 1992), 1936--1948.
\newblock
Issue 12.


\bibitem[\protect\citeauthoryear{Tassiulas and Ephremides}{Tassiulas and
  Ephremides}{1993}]%
        {TaE93}
{Leandros Tassiulas} {and} {Anthony Ephremides}. 1993.
\newblock \showarticletitle{{Dynamic Server Allocation to Parallel Queues with
  Randomly Varying Connectivity}}.
\newblock {\em IEEE Transactions on Information Theory\/} {39}, 2 (March 1993),
  466--478.
\newblock


\bibitem[\protect\citeauthoryear{van~de Ven, Borst, and Ying}{van~de Ven
  et~al\mbox{.}}{2013}]%
        {VBY13}
{P.M. van~de Ven}, {S.C. Borst}, {and} {L. Ying}. 2013.
\newblock \showarticletitle{{Inefficiency of MaxWeight scheduling in spatial
  wireless networks}}.
\newblock {\em Computer Communications\/}  {36} (July 2013), 1350--1359.
\newblock
Issue 12.


\end{thebibliography}



\begin{thebibliography}{00}


\ifx \showCODEN    \undefined \def \showCODEN     #1{\unskip}     \fi
\ifx \showDOI      \undefined \def \showDOI       #1{{\tt DOI:}\penalty0{#1}\ }
  \fi
\ifx \showISBNx    \undefined \def \showISBNx     #1{\unskip}     \fi
\ifx \showISBNxiii \undefined \def \showISBNxiii  #1{\unskip}     \fi
\ifx \showISSN     \undefined \def \showISSN      #1{\unskip}     \fi
\ifx \showLCCN     \undefined \def \showLCCN      #1{\unskip}     \fi
\ifx \shownote     \undefined \def \shownote      #1{#1}          \fi
\ifx \showarticletitle \undefined \def \showarticletitle #1{#1}   \fi
\ifx \showURL      \undefined \def \showURL       #1{#1}          \fi

\bibitem[\protect\citeauthoryear{Ablamowicz and Fauser}{Ablamowicz and
  Fauser}{2007}]%
        {Ablamowicz07}
{Rafal Ablamowicz} {and} {Bertfried Fauser}. 2007.
\newblock CLIFFORD: a Maple 11 Package for Clifford Algebra Computations,
  version 11.
\newblock   (2007).
\newblock
\showURL{%
Retrieved February 28, 2008 from
  \url{http://math.tntech.edu/rafal/cliff11/index.html}}


\bibitem[\protect\citeauthoryear{Abril and Plant}{Abril and Plant}{2007}]%
        {Abril07}
{Patricia~S. Abril} {and} {Robert Plant}. 2007.
\newblock \showarticletitle{The patent holder's dilemma: Buy, sell, or troll?}
\newblock {\it Commun. ACM} {50}, 1 (Jan. 2007), 36--44.
\newblock
\showDOI{%
\url{http://dx.doi.org/10.1145/1188913.1188915}}


\bibitem[\protect\citeauthoryear{Andler}{Andler}{1979}]%
        {Andler79}
{Sten Andler}. 1979.
\newblock \showarticletitle{Predicate Path expressions}. In {\em Proceedings of
  the 6th. ACM SIGACT-SIGPLAN symposium on Principles of Programming Languages}
  {\em (POPL '79)}. ACM Press, New York, NY, 226--236.
\newblock
\showDOI{%
\url{http://dx.doi.org/10.1145/567752.567774}}


\bibitem[\protect\citeauthoryear{Anisi}{Anisi}{2003}]%
        {anisi03}
{David~A. Anisi}. 2003.
\newblock {\em Optimal Motion Control of a Ground Vehicle}.
\newblock Master's\ thesis. Royal Institute of Technology (KTH), Stockholm,
  Sweden.
\newblock


\bibitem[\protect\citeauthoryear{Clarkson}{Clarkson}{1985}]%
        {Clarkson85}
{Kenneth~L. Clarkson}. 1985.
\newblock {\em Algorithms for Closest-Point Problems (Computational Geometry)}.
\newblock Ph.D. Dissertation. Stanford University, Palo Alto, CA.
\newblock
\newblock
\shownote{UMI Order Number: AAT 8506171.}


\bibitem[\protect\citeauthoryear{Cohen}{Cohen}{1996}]%
        {JCohen96}
{Jacques Cohen} (Ed.). 1996.
\newblock \showarticletitle{Special Issue: Digital Libraries}.
\newblock {\em Commun. {ACM}\/} {39}, 11 (Nov. 1996).
\newblock


\bibitem[\protect\citeauthoryear{Cohen, Nutt, and Sagic}{Cohen
  et~al\mbox{.}}{2007}]%
        {Cohen07}
{Sarah Cohen}, {Werner Nutt}, {and} {Yehoshua Sagic}. 2007.
\newblock \showarticletitle{Deciding equivalances among conjunctive aggregate
  queries}.
\newblock {\em J. ACM\/} {54}, 2, Article 5 (April 2007), 50 pages.
\newblock
\showDOI{%
\url{http://dx.doi.org/10.1145/1219092.1219093}}


\bibitem[\protect\citeauthoryear{Douglass, Harel, and Trakhtenbrot}{Douglass
  et~al\mbox{.}}{1998}]%
        {Douglass98}
{Bruce~P. Douglass}, {David Harel}, {and} {Mark~B. Trakhtenbrot}. 1998.
\newblock \showarticletitle{Statecarts in use: structured analysis and
  object-orientation}.
\newblock In {\em Lectures on Embedded Systems}, {Grzegorz Rozenberg} {and}
  {Frits~W. Vaandrager} (Eds.). Lecture Notes in Computer Science, Vol. 1494.
  Springer-Verlag, London, 368--394.
\newblock
\showDOI{%
\url{http://dx.doi.org/10.1007/3-540-65193-4_29}}


\bibitem[\protect\citeauthoryear{Editor}{Editor}{2007}]%
        {Editor00}
{Ian Editor} (Ed.). 2007.
\newblock {\em The title of book one\/} (1st. ed.). The name of the series one,
  Vol.~9.
\newblock University of Chicago Press, Chicago.
\newblock
\showDOI{%
\url{http://dx.doi.org/10.1007/3-540-09237-4}}


\bibitem[\protect\citeauthoryear{Editor}{Editor}{2008}]%
        {Editor00a}
{Ian Editor} (Ed.). 2008.
\newblock {\em The title of book two\/} (2nd. ed.).
\newblock University of Chicago Press, Chicago, Chapter 100.
\newblock
\showDOI{%
\url{http://dx.doi.org/10.1007/3-540-09237-4}}


\bibitem[\protect\citeauthoryear{Gundy, Balzarotti, and Vigna}{Gundy
  et~al\mbox{.}}{2007}]%
        {VanGundy07}
{Matthew~Van Gundy}, {Davide Balzarotti}, {and} {Giovanni Vigna}. 2007.
\newblock \showarticletitle{Catch me, if you can: Evading network signatures
  with web-based polymorphic worms}. In {\em Proceedings of the first USENIX
  workshop on Offensive Technologies} {\em (WOOT '07)}. USENIX Association,
  Berkley, CA, Article 7, 9 pages.
\newblock


\bibitem[\protect\citeauthoryear{Harel}{Harel}{1978}]%
        {Harel78}
{David Harel}. 1978.
\newblock {\em LOGICS of Programs: AXIOMATICS and DESCRIPTIVE POWER}.
\newblock MIT Research Lab Technical Report TR-200. Massachusetts Institute of
  Technology, Cambridge, MA.
\newblock


\bibitem[\protect\citeauthoryear{Harel}{Harel}{1979}]%
        {Harel79}
{David Harel}. 1979.
\newblock {\em First-Order Dynamic Logic}. Lecture Notes in Computer Science,
  Vol.~68.
\newblock Springer-Verlag, New York, NY.
\newblock
\showDOI{%
\url{http://dx.doi.org/10.1007/3-540-09237-4}}


\bibitem[\protect\citeauthoryear{H{\"o}rmander}{H{\"o}rmander}{1985a}]%
        {MR781537}
{Lars H{\"o}rmander}. 1985a.
\newblock {\em The analysis of linear partial differential operators. {III}}.
  Grundlehren der Mathematischen Wissenschaften [Fundamental Principles of
  Mathematical Sciences], Vol. 275.
\newblock Springer-Verlag, Berlin, Germany. viii+525 pages.
\newblock
\showISBNx{3-540-13828-5}
\newblock
\shownote{Pseudodifferential operators.}


\bibitem[\protect\citeauthoryear{H{\"o}rmander}{H{\"o}rmander}{1985b}]%
        {MR781536}
{Lars H{\"o}rmander}. 1985b.
\newblock {\em The analysis of linear partial differential operators. {IV}}.
  Grundlehren der Mathematischen Wissenschaften [Fundamental Principles of
  Mathematical Sciences], Vol. 275.
\newblock Springer-Verlag, Berlin, Germany. vii+352 pages.
\newblock
\showISBNx{3-540-13829-3}
\newblock
\shownote{Fourier integral operators.}


\bibitem[\protect\citeauthoryear{Kirschmer and Voight}{Kirschmer and
  Voight}{2010}]%
        {Kirschmer:2010:AEI:1958016.1958018}
{Markus Kirschmer} {and} {John Voight}. 2010.
\newblock \showarticletitle{Algorithmic Enumeration of Ideal Classes for
  Quaternion Orders}.
\newblock {\em SIAM J. Comput.\/} {39}, 5 (Jan. 2010), 1714--1747.
\newblock
\showISSN{0097-5397}
\showDOI{%
\url{http://dx.doi.org/10.1137/080734467}}


\bibitem[\protect\citeauthoryear{Knuth}{Knuth}{1997}]%
        {Knuth97}
{Donald~E. Knuth}. 1997.
\newblock {\em The Art of Computer Programming, Vol. 1: Fundamental Algorithms
  (3rd. ed.)}.
\newblock Addison Wesley Longman Publishing Co., Inc.
\newblock


\bibitem[\protect\citeauthoryear{Kosiur}{Kosiur}{2001}]%
        {Kosiur01}
{David Kosiur}. 2001.
\newblock {\em Understanding Policy-Based Networking\/} (2nd. ed.).
\newblock Wiley, New York, NY.
\newblock


\bibitem[\protect\citeauthoryear{Lee}{Lee}{2005}]%
        {Lee05}
{Newton Lee}. 2005.
\newblock \showarticletitle{Interview with Bill Kinder: January 13, 2005}.
\newblock Video, {\em Comput. Entertain.\/} {3}, 1, Article 4 (Jan.-March
  2005).
\newblock
\showDOI{%
\url{http://dx.doi.org/10.1145/1057270.1057278}}


\bibitem[\protect\citeauthoryear{Novak}{Novak}{2003}]%
        {Novak03}
{Dave Novak}. 2003.
\newblock \showarticletitle{Solder man}. Video. In {\em ACM SIGGRAPH 2003 Video
  Review on Animation theater Program: Part I - Vol. 145 (July 27--27, 2003)}.
  ACM Press, New York, NY, 4.
\newblock
\showDOI{%
\url{http://dx.doi.org/99.9999/woot07-S422}}


\bibitem[\protect\citeauthoryear{Obama}{Obama}{2008}]%
        {Obama08}
{Barack Obama}. 2008.
\newblock A more perfect union.
\newblock Video.   (5 March 2008).
\newblock
\showURL{%
Retrieved March 21, 2008 from
  \url{http://video.google.com/videoplay?docid=6528042696351994555}}


\bibitem[\protect\citeauthoryear{Poker-Edge.Com}{Poker-Edge.Com}{2006}]%
        {Poker06}
{Poker-Edge.Com}. 2006.
\newblock Stats and Analysis.
\newblock   (March 2006).
\newblock
\showURL{%
Retrieved June 7, 2006 from \url{http://www.poker-edge.com/stats.php}}


\bibitem[\protect\citeauthoryear{Rous}{Rous}{2008}]%
        {rous08}
{Bernard Rous}. 2008.
\newblock \showarticletitle{The Enabling of Digital Libraries}.
\newblock {\em Digital Libraries\/} {12}, 3, Article 5 (July 2008).
\newblock
\newblock
\shownote{To appear.}


\bibitem[\protect\citeauthoryear{Saeedi, Zamani, and Sedighi}{Saeedi
  et~al\mbox{.}}{2010a}]%
        {SaeediMEJ10}
{Mehdi Saeedi}, {Morteza~Saheb Zamani}, {and} {Mehdi Sedighi}. 2010a.
\newblock \showarticletitle{A library-based synthesis methodology for
  reversible logic}.
\newblock {\em Microelectron. J.\/} {41}, 4 (April 2010), 185--194.
\newblock


\bibitem[\protect\citeauthoryear{Saeedi, Zamani, Sedighi, and Sasanian}{Saeedi
  et~al\mbox{.}}{2010b}]%
        {SaeediJETC10}
{Mehdi Saeedi}, {Morteza~Saheb Zamani}, {Mehdi Sedighi}, {and} {Zahra
  Sasanian}. 2010b.
\newblock \showarticletitle{Synthesis of Reversible Circuit Using Cycle-Based
  Approach}.
\newblock {\em J. Emerg. Technol. Comput. Syst.\/} {6}, 4 (Dec. 2010).
\newblock


\bibitem[\protect\citeauthoryear{Scientist}{Scientist}{2009}]%
        {JoeScientist001}
{Joseph Scientist}. 2009.
\newblock The fountain of youth.
\newblock   (Aug. 2009).
\newblock
\newblock
\shownote{Patent No. 12345, Filed July 1st., 2008, Issued Aug. 9th., 2009.}


\bibitem[\protect\citeauthoryear{Smith}{Smith}{2010}]%
        {Smith10}
{Stan~W. Smith}. 2010.
\newblock \showarticletitle{An experiment in bibliographic mark-up: Parsing
  metadata for XML export}. In {\em Proceedings of the 3rd. annual workshop on
  Librarians and Computers} {\em (LAC '10)}, {Reginald~N. Smythe} {and}
  {Alexander Noble} (Eds.), Vol.~3. Paparazzi Press, Milan Italy, 422--431.
\newblock
\showDOI{%
\url{http://dx.doi.org/99.9999/woot07-S422}}


\bibitem[\protect\citeauthoryear{Spector}{Spector}{1990}]%
        {Spector90}
{Asad~Z. Spector}. 1990.
\newblock \showarticletitle{Achieving application requirements}.
\newblock In {\em Distributed Systems} (2nd. ed.), {Sape Mullender} (Ed.). ACM
  Press, New York, NY, 19--33.
\newblock
\showDOI{%
\url{http://dx.doi.org/10.1145/90417.90738}}


\bibitem[\protect\citeauthoryear{Thornburg}{Thornburg}{2001}]%
        {Thornburg01}
{Harry Thornburg}. 2001.
\newblock Introduction to Bayesian Statistics.
\newblock   (March 2001).
\newblock
\showURL{%
Retrieved March 2, 2005 from
  \url{http://ccrma.stanford.edu/~jos/bayes/bayes.html}}


\end{thebibliography}
\bibliographystyle{ACM-Reference-Format-Journals}

\section{Appendix}

\subsection{Proof of Lemma \ref{lemma_R_in_close_C}}
\label{appendix_pf_lemma_F_in_close_C}
Given $\reqvec \in \feasibilityRegion$, since it is feasible there exists a user prioritization policy $\eta$ that fulfills $\reqvec$, i.e., the Markov chain $\{\mathbf{X}(t)\}_{t\geq 1}$ is positive recurrent, which implies there exists a stationary distribution over the state space. Since $\eta$ is a stationary policy that picks $\mathbf{d}(t+1)$ based on $\mathbf{X}(t)$, by Ergodic Theorem each priority decision $\mathbf{d}$ is selected with some time fraction $\alpha_\mathbf{d}$ such that $\sum\limits_{\mathbf{d} \in D}\alpha_\mathbf{d} = 1$ . If we consider $X_i(t)$ as a queue, the average arrival $\reqscalar_i$  should not be bigger than the average departure which is given by $\sum\limits_{\mathbf{d}\in D} \alpha_\mathbf{d} p_i(\mathbf{d})$ since otherwise $X_i(t)$ goes a.s. to infinity and the chain cannot be positive recurrent. 

Therefore, 
$$
\reqvec \preceq \sum\limits_{\mathbf{d}\in D} \alpha_\mathbf{d} \mathbf{p}(\mathbf{d}) \in \text{Conv}(P), 
$$
which implies $\reqvec \in C \subseteq \text{cl}(C)$. 

\subsection{Proof of Theorem \ref{theorem_MW_feasibility_optimal}}
\label{appendix_pf_thm_MW_feasibility_optimal}
We start by introducing a lemma. 

\begin{lemma}
\label{lemma_feasibility_region_projection_understanding}
A payoff requirement vector $\reqvec$ is in $C$ if and only if, for any non-negative vector $\boldsymbol{\gamma} \succeq 0$, there exists a priority decision $\mathbf{d}$, such that
$$
\langle \boldsymbol{\gamma}, \reqvec \rangle \leq \langle \boldsymbol{\gamma}, \mathbf{p}(\mathbf{d}) \rangle. 
$$

Equivalently, $\reqvec$ is in $C$ if and only if, for any $\boldsymbol{\gamma} \succeq \mathbf{0}$, 
$$
\langle \boldsymbol{\gamma}, \reqvec \rangle \leq \max\limits_{\mathbf{d} \in D} \langle \boldsymbol{\gamma}, \mathbf{p}(\mathbf{d}) \rangle. 
$$
\end{lemma}

To understand this lemma, since any vector in $C$ is dominated by some $\reqvec \in \text{Conv}(P)$, we consider for simplicity a $\reqvec \in \text{Conv}(P)$. Such a vector can be expressed as a convex combination of expected payoff vectors, i.e., $\reqvec = \sum\limits_{\mathbf{d} \in D} \alpha_{\mathbf{d}} \mathbf{p}(\mathbf{d})$, where 
$\sum\limits_{\mathbf{d} \in D} \alpha_{\mathbf{d}} = 1 $ and 
$\alpha_{\mathbf{d}} \geq 0, \forall \mathbf{d} \in D. $

Now we have
\begin{align*}
\langle \boldsymbol{\gamma}, \reqvec \rangle = & \sum\limits_{\mathbf{d} \in D} \alpha_{\mathbf{d}} \langle \boldsymbol{\gamma}, \mathbf{p}(\mathbf{d}) \rangle \\
\leq & \sum\limits_{\mathbf{d} \in D} \alpha_{\mathbf{d}} \max\limits_{\mathbf{d} \in D} \langle \boldsymbol{\gamma}, \mathbf{p}(\mathbf{d}) \rangle	\\
= & \max\limits_{\mathbf{d} \in D} \langle \boldsymbol{\gamma}, \mathbf{p}(\mathbf{d}) \rangle. 
\end{align*}

This actually proves the necessity of the condition. The formal proof is shown below. 

\begin{proof}[Proof of Lemma \ref{lemma_feasibility_region_projection_understanding}]
Given a payoff requirement vector $\reqvec$, by definition of $C$, we know that $\reqvec$ lying in set $C$ is equivalent to the {\em feasibility} of the following set of linear equations and inequalities, 

\begin{align}
\label{align_feasible_set}
\left\{ 
   \begin{array}{l}
     \reqvec \preceq \sum\limits_{\mathbf{d} \in D} \alpha_{\mathbf{d}} \mathbf{p}(\mathbf{d}) \\
     \alpha_{\mathbf{d}} \geq 0, \forall \mathbf{d} \in D	\\
     \sum\limits_{\mathbf{d} \in D} \alpha_{\mathbf{d}} = 1. 
   \end{array} 
\right.
\end{align}

The condition in Lemma \ref{lemma_feasibility_region_projection_understanding} is equivalent to the {\em infeasibility} of 
\begin{align}
\label{align_infeasible_set}
\left\{ 
   \begin{array}{l}
     \boldsymbol{\gamma} \succeq \mathbf{0} \\
     \langle \boldsymbol{\gamma},\reqvec \rangle > \langle \boldsymbol{\gamma}, \mathbf{p}(\mathbf{d}) \rangle, \forall \mathbf{d} \in D. 
   \end{array} 
\right.
\end{align}

By strong duality \cite{BoV09} it is easy to prove that set (\ref{align_feasible_set}) being feasible is equivalent to set (\ref{align_infeasible_set}) being infeasible and this concludes the proof. 
\end{proof}

\begin{proof}[Proof of Theorem \ref{theorem_MW_feasibility_optimal}]
By Lemma \ref{lemma_R_in_close_C} we know $F \subseteq \text{cl}(C)$ and by definition we know $\feasibilityRegion_{\text{MW}} \subseteq \feasibilityRegion$. To prove the theorem it suffices to show $\text{int}(C) \subseteq \feasibilityRegion_\text{MW}$. We show this by constructing a Lyapunov function and using Foster's theorem. 

Given $\reqvec \in \text{int}(C)$, the goal is to show $\reqvec$ can be fulfilled by the MW policy. 

By definition of interior there exists $\epsilon > 0$ such that $\reqvec^\prime = \reqvec + \epsilon \mathbf{1} \in C$ where $\mathbf{1} = (1, 1, \cdots, 1)$. We define a Lyapunov function as 
$$
L(\mathbf{X}(t)) = \sum\limits_{i=1}^{n} X_i(t)^2. 
$$

In period $t+1$, we have
\begin{eqnarray}
\lefteqn{ \expectation\left[ L(\mathbf{X}(t+1)) - L(\mathbf{X}(t)) | \mathbf{X}(t) = \mathbf{x} \right]} \notag \\
& = & \expectation\Bigg[\sum\limits_{i=1}^{n} X_i(t+1)^2 - X_i(t)^2 | \mathbf{X}(t) = \mathbf{x}\Bigg]	\notag	\\
& \leq & \expectation\Bigg[\sum\limits_{i=1}^{n} (X_i(t) + \reqscalar_i \notag - V_i(\mathbf{d}(t+1)))^2 - X_i(t)^2 | \mathbf{X}(t) = \mathbf{x}\Bigg] \notag \\
& = & \expectation\Bigg[\sum\limits_{i=1}^{n} (\reqscalar_i - V_i(\mathbf{d}(t+1)))^2 + 2 \langle \mathbf{X}(t), \reqvec - \mathbf{V}(\mathbf{d}(t+1)) \rangle | \mathbf{X}(t) = \mathbf{x}\Bigg] 	\notag \\
& \leq & \expectation\Bigg[\sum\limits_{i=1}^{n} (\reqscalar_i^2 + V_i(\mathbf{d}(t+1))^2) + 2\langle \mathbf{X}(t), \reqvec - \mathbf{V}(\mathbf{d}(t+1)) \rangle | \mathbf{X}(t) = \mathbf{x}\Bigg] \label{align_drift}
\end{eqnarray}

Under the MW policy, by (\ref{align_max_weight}) we know that
\begin{eqnarray*}
\lefteqn{\expectation[\langle \mathbf{X}(t), \mathbf{V}(\mathbf{d}(t+1)) \rangle | \mathbf{X}(t) = \mathbf{x}]}	\\ 
& = & \expectation[\langle \mathbf{X}(t), \mathbf{p}(\mathbf{d}(t+1)) \rangle | \mathbf{X}(t) = \mathbf{x}]	\\
& = & \max\limits_{\mathbf{d} \in D} \langle \mathbf{x}, \mathbf{p}(\mathbf{d}) \rangle. 
\end{eqnarray*}

Since $\reqvec^\prime =  \reqvec + \epsilon \mathbf{1} \in C$, we get that 
\begin{eqnarray*}
\lefteqn{\expectation[\langle \mathbf{X}(t), \reqvec - \mathbf{V}(\mathbf{d}(t+1)) \rangle | \mathbf{X}(t) = \mathbf{x}]}	\\
& = & \expectation[\langle \mathbf{X}(t), \reqvec^\prime - \mathbf{V}(\mathbf{d}(t+1)) \rangle | \mathbf{X}(t) = \mathbf{x}] - \epsilon \langle \mathbf{x}, \mathbf{1}\rangle	\\
& = & \langle \mathbf{x}, \reqvec^\prime \rangle - \max\limits_{\mathbf{d} \in D} \langle \mathbf{x}, \mathbf{p}(\mathbf{d})\rangle - \epsilon\langle \mathbf{x}, \mathbf{1}\rangle	\\
& \leq & - \epsilon\langle \mathbf{x}, \mathbf{1}\rangle, 
\end{eqnarray*}
where the last step is true by Lemma \ref{lemma_feasibility_region_projection_understanding}. 

Also since there are finite payoff vectors, we use $b_1$ to represent an upper bound on all possible payoff values and payoff requirements. 
Therefore, by (\ref{align_drift}) we get that
\begin{align*}
\expectation[L(\mathbf{X}(t+1)) - L(\mathbf{X}(t)) | \mathbf{X}(t) = \mathbf{x}] & \leq 2nb_1^2 - 2 \epsilon \langle \mathbf{x}, \mathbf{1} \rangle \\
& \leq -1
\end{align*}
for $\mathbf{x}$ satisfying $\langle\mathbf{x}, \mathbf{1}\rangle \geq \frac{nb_1^2}{\epsilon} + \frac{1}{2\epsilon}$. 

It is not hard to show\footnote{This is true because given our assumption that requirement $\reqvec$ and the payoff vectors are rational valued and have finite options, the state space of process $\{\mathbf{X}(t)\}_{t\geq 1}$ is in a lattice, see e.g., \cite{CoS13b}.} there are finite states $\mathbf{x}$ with $\langle\mathbf{x}, \mathbf{1}\rangle < \frac{nb_1^2}{\epsilon} + \frac{1}{2\epsilon}$. Therefore, by Foster's theorem, $\{\mathbf{X}(t)\}_{t\geq 1}$ is positive recurrent and $\reqvec$ is fulfilled by the MW policy. 
Thus, this shows that $\text{int}(C) \subseteq \feasibilityRegion_\text{MW}$. 
\end{proof}

\subsection{Proof of Theorem \ref{thm_R_IB}}
\label{appendix_pf_thm_R_IB_chap2}
We first introduce some further notation. Given two vectors $\boldsymbol{a}=(a_1, a_2, \cdots, a_n)$ and $\mathbf{b} = (b_1, b_2, \cdots, b_n)$, we denote by $\boldsymbol{a} \circ \mathbf{b}=(a_1 b_1, a_2 b_2, \cdots, a_n b_n)$ the entrywise product. 

For any $\mathbf{w} \succ \mathbf{0}$ and $\reqvec \in \text{int}(R_\text{IB})$, the goal is to show that $\reqvec$ can be fulfilled by the $\mathbf{w}$-LDF policy. 

Let $\boldsymbol{\beta} = (\frac{1}{w_1}, \frac{1}{w_2}, \cdots, \frac{1}{w_n})$. 
By definition of interior there exists an $\epsilon>0$ such that ${\reqvec}^\prime = \reqvec + \epsilon\boldsymbol{\beta} \in R_{\text{IB}}$. 
By definition of $R_{\text{IB}}$, there exists a vector $\ribvec \succ \mathbf{0}$ such that $\reqvec^\prime$ and $\ribvec$ satisfy the conditions (\ref{align_R_IB_defn}). 

Consider the following candidate Lyapunov function: 
$$
L(\mathbf{X}(t)) = \sum\limits_{i=1}^{n} {\ribscalar}_i w_i X_i(t)^2. 
$$

Note that we consider a process $\{\mathbf{X}(t)\}_{t\geq 1}$ that is driven by the $\mathbf{w}$-LDF policy. In period $t+1$, by similar analysis as in (\ref{align_drift}) we have that
\begin{eqnarray}
\lefteqn{\expectation\left[L(\mathbf{X}(t+1)) - L(\mathbf{X}(t)) | \mathbf{X}(t) = \mathbf{x}\right]} \notag \\
& \leq & \expectation\Bigg[\sum\limits_{i=1}^{n} {\ribscalar}_i w_i (\reqscalar_i^2 + V_i(\mathbf{d}(t+1))^2) + 2 \langle \ribvec \circ \mathbf{w} \circ \mathbf{X}(t), \reqvec - \mathbf{V}(\mathbf{d}(t+1)) \rangle | \mathbf{X}(t) = \mathbf{x} \Bigg] 	\notag \\ 
& = & \expectation \Bigg[ \sum\limits_{i=1}^{n} {\ribscalar}_i w_i (\reqscalar_i^2 + V_i(\mathbf{d}(t+1))^2) + 2\langle \ribvec \circ \mathbf{w} \circ \mathbf{X}(t), {\reqvec}^\prime - \mathbf{V}(\mathbf{d}(t+1)) \rangle | \mathbf{X}(t) = \mathbf{x}\Bigg]	\notag	\\
& & - 2\epsilon\langle \mathbf{x}, \ribvec \rangle	\label{align_LDF_drift_chap2}
\end{eqnarray}

Let $\mathbf{d}$ denote the priority decision selected according to $\mathbf{w}$-LDF policy. Thus, We have that
\begin{eqnarray*}
& &\expectation\left[ \langle \ribvec \circ \mathbf{w} \circ \mathbf{X}(t), {\reqvec}^\prime - \mathbf{V}(\mathbf{d}(t+1)) \rangle | \mathbf{X}(t) = \mathbf{x} \right]\\
& &~~ = \langle \ribvec \circ \mathbf{w} \circ \mathbf{x}, {\reqvec}^\prime - \mathbf{p}(\mathbf{d}) \rangle. 
\end{eqnarray*}
By reordering users according to priorities, we get
\begin{eqnarray*}
\lefteqn{\langle \ribvec \circ \mathbf{w} \circ \mathbf{x}, {\reqvec}^\prime - \mathbf{p}(\mathbf{d}) \rangle}	\\
& = & \sum\limits_{i=1}^{n} w_{d_i} x_{d_i} [{\ribscalar}_{d_i} {\reqscalar}^\prime_{d_i} - {\ribscalar}_{d_i} p_{d_i}(\mathbf{d})]	\\
& = & \sum\limits_{i=1}^{n-1} [w_{d_i} x_{d_i} - w_{d_{i+1}} x_{d_{i+1}}] [\sum\limits_{j=1}^i {\ribscalar}_{d_j} {\reqscalar}^\prime_{d_j} - \sum\limits_{j=1}^i \ribscalar_{d_j} p_{d_j}(\mathbf{d})]	\\
& & + w_{d_n} x_{d_n} [\sum\limits_{j=1}^n {\ribscalar}_{d_j} {\reqscalar}^\prime_{d_j} - \sum\limits_{j=1}^n \ribscalar_{d_j} p_{d_j}(\mathbf{d})]. 
\end{eqnarray*}

By $\mathbf{w}$-LDF policy we know $w_{d_i} x_{d_i} \geq w_{d_{i+1}} x_{d_{i+1}}$. By (\ref{align_R_IB_defn}) we have $\sum\limits_{j=1}^i {\ribscalar}_{d_j} {\reqscalar}^\prime_{d_j} \leq \sum\limits_{j=1}^i \ribscalar_{d_j} p_{d_j}(\mathbf{d})$ for $1 \leq i \leq n$. Therefore, 

$$
\expectation\left[\langle \ribvec \circ \mathbf{w} \circ \mathbf{X}(t), {\reqvec}^\prime - \mathbf{V}(\mathbf{d}(t+1)) \rangle | \mathbf{X}(t) = \mathbf{x} \right] \leq 0. 
$$

Suppose $b_2$ is an upper bound on all ${\ribscalar}_i$ and $w_i$, by (\ref{align_LDF_drift_chap2}), we get that
\begin{align*}
\expectation[L(\mathbf{X}(t+1)) - L(\mathbf{X}(t)) | \mathbf{X}(t) = \mathbf{x}] \leq & 2n b_2^2 b_1^2 - 2 \epsilon \langle \mathbf{x}, \ribvec \rangle \\
\leq & -1
\end{align*}
for $\mathbf{x}$ satisfying $\langle\mathbf{x}, \ribvec\rangle \geq \frac{n b_2^2 b_1^2}{\epsilon} + \frac{1}{2\epsilon}$. 

Again, since there are finite states $\mathbf{x}$ with $\langle\mathbf{x}, \ribvec\rangle < \frac{n b_2^2 b_1^2}{\epsilon} + \frac{1}{2\epsilon}$, by Foster's Theorem $\{\mathbf{X}(t)\}_{t\geq 1}$ is positive recurrent and $\reqvec$ is fulfilled by the $\mathbf{w}$-LDF policy. 

Therefore, for any $\mathbf{w} \succ \mathbf{0}$, we have that
$$
\text{int}(R_\text{IB}) \subseteq \feasibilityRegion_{\mathbf{w}\text{-LDF}}. 
$$

\subsection{Proof of Theorem \ref{thm_visualize_R_IB}}
\label{appendix_pf_thm_visualize_R_IB}
The proof of Theorem \ref{thm_visualize_R_IB} is complicated and we apologize for that. Part of the complication comes from understanding how the property of monotonicity in payoffs characterizes the geometry of the region $R_\text{IB}$. Further, a feasible payoff requirement implies feasibilities for all user subsets, and thus we need to look at the projections in all subspaces. 

Figure~\ref{fig_outline_pf_thm_3} gives the high-level outline for the proof of Theorem \ref{thm_visualize_R_IB}. There are two parts which involve the technical results Lemma \ref{lemma_q_S_in_C_S} and \ref{lemma_finding_nonnegative_beta}, which we will state in the proof. In order to allow the reader follow the proof, we defer their own proof to later. 

\begin{figure}[htp]
  \centering
  \includegraphics[width=0.85\textwidth]{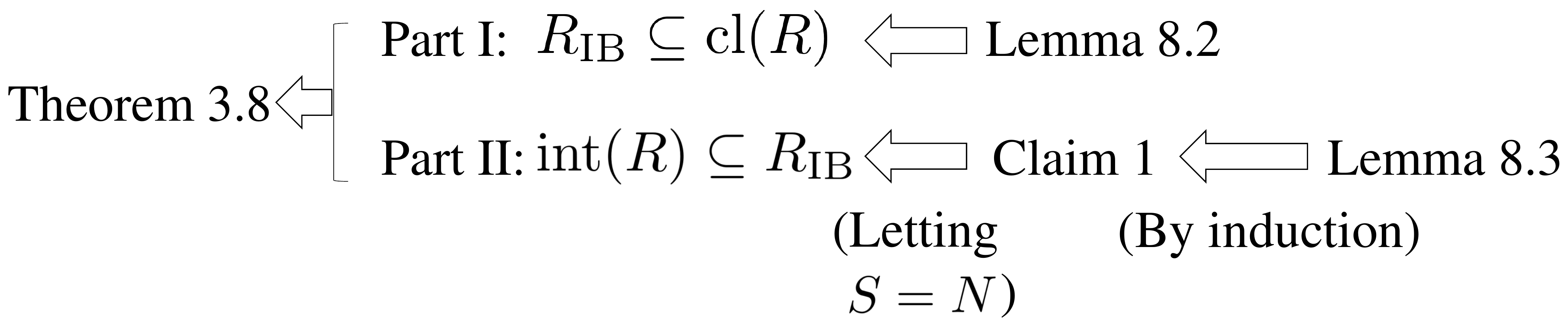}
  \caption{Outline for the proof of Theorem \ref{thm_visualize_R_IB}. }
  \label{fig_outline_pf_thm_3}
\end{figure}

\begin{lemma}
\label{lemma_q_S_in_C_S}
If a system satisfies monotonicity in payoffs, then for all $\reqvec \in C$ and all subsets of users $S\subseteq \fullUserSet$, $\reqvec^S \in C^S$. 
\end{lemma}

We have argued in Section \ref{subsection_characterizing_R_IB} that $C^S$ does not necessarily equal to the projection of region $C$ on the subspace of $S$ in general, but the statement is true if the system satisfies monotonicity in payoffs. Please refer to Appendix \ref{pf_lemma_q_S_in_C_S} for detailed proof of this lemma. 

Given a subset of users $S \subseteq \fullUserSet$, for two vectors $\mathbf{p}^S$ and $\mathbf{q}^S$, we say $\mathbf{p}^S \succS{S} \mathbf{q}^S$ if $p_i^S > q_i^S$ for any $i \in S$. 
We define the subspace of $S$ as $\mathbb R_+^S = \{\reqvec\in \mathbb R_+^n | q_i = 0, \forall i\notin S \}$. 

For a region $X\subseteq \mathbb R_+^n$ which lies in the subspace of $S$, we denote by $\text{int}^S(X)$ the relative interior of $X$, i.e., the interior of $X$, relative to the subspace of $S$. 
Similarly we use $\text{cl}^S(X)$, $\text{bd}^S(X)$ to represent the relative closure and boundary of $X$ in the subspace of $S$, respectively. 

Following the definition of region $\concaveHull$ in Definition \ref{defn_concave_hull}, we can express int$(\concaveHull)$ and cl$(\concaveHull)$ as below, 
\begin{align}
\text{int}(\concaveHull) & = \{ \reqvec | \forall S \subseteq \{1, 2, \cdots, n\}, \reqvec^S \in \text{int}^S(C^S \setminus B^S) \}, \label{align_int_R}	\\
\text{cl}(\concaveHull) & = \{ \reqvec | \forall S \subseteq \{1, 2, \cdots, n\}, \reqvec^S \in \text{cl}^S(C^S \setminus B^S) \}. \label{align_cl_R}
\end{align}
By definition, we know $C^S$ and $B^S$ are closed sets. We can also show\footnote{This is clear from the definition of $C^S$ and $B^S$. We omit the proof to save space. }
\begin{align*}
\text{int}^S(C^S \setminus B^S) & = \text{int}^S(C^S) \setminus B^S, 	\\
\text{cl}^S(C^S \setminus B^S) & = C^S \setminus \text{int}^S(B^S). 
\end{align*}

Next we prove Theorem \ref{thm_visualize_R_IB} in two parts: $R_{\text{IB}} \subseteq \text{cl}(\concaveHull)$ and $\text{int}(\concaveHull) \subseteq R_{\text{IB}}$. 

{\bf Part I of the proof: }

We start with the easy part and first show $R_{\text{IB}} \subseteq \text{cl}(\concaveHull)$.

Given $\mathbf{q} \in R_{\text{IB}}$, by definition, there exists $\boldsymbol{\alpha} \succ \mathbf{0}$ such that for all subsets of users $S$, 
$$
\sum\limits_{i\in S}\alpha_i q_i \leq \min\limits_{\mathbf{d} \in D(S)} \sum\limits_{i\in S} \alpha_i p_i(\mathbf{d}), 
$$
i.e., 
\begin{equation}
\label{equ_alpha_S_q_S}
\langle \boldsymbol{\alpha}^S, \mathbf{q}^S \rangle \leq \min\limits_{\mathbf{d} \in D(S)} \langle \boldsymbol{\alpha}^S, \mathbf{p}^S(\mathbf{d}) \rangle. 
\end{equation}

To show $\mathbf{q} \in \text{cl}(\concaveHull)$, by (\ref{align_cl_R}) we need to show for all subsets of users $S$ that $\mathbf{q}^S \in C^S \setminus \text{int}^S(B^S)$. 

Since $\mathbf{q} \in R_{\text{IB}} \subseteq C$, by Lemma \ref{lemma_q_S_in_C_S} we have $\mathbf{q}^S \in C^S$. 

Now suppose for some subset of users $S$, $\mathbf{q}^S \in \text{int}^S(B^S)$. By definition of $B^S$ and interior, there exists $\mathbf{x}^S \in \text{Conv}(P^S)$ such that $\mathbf{q}^S \succeq \mathbf{x}^S$ and $\mathbf{q}^S \neq \mathbf{x}^S$, which implies that
$$
\langle \boldsymbol{\alpha}^S, \mathbf{q}^S \rangle > \langle \boldsymbol{\alpha}^S, \mathbf{x}^S \rangle. 
$$

\newcommand*{\cvxCombPar}{c}

Since $\mathbf{x}^S \in \text{Conv}(P^S)$, there exists $\{\cvxCombPar_\mathbf{d} | \mathbf{d}\in D(S)\}$ such that $\mathbf{x}^S = \sum\limits_{\mathbf{d} \in D(S)} \cvxCombPar_\mathbf{d} \mathbf{p}^S(\mathbf{d})$ and $\cvxCombPar_\mathbf{d} \geq 0, \sum\limits_{\mathbf{d} \in D(S)} \cvxCombPar_\mathbf{d} = 1$. Therefore, 
\begin{align}
\label{align_cvx_greater_than_min}
\langle \boldsymbol{\alpha}^S, \mathbf{q}^S \rangle > \langle \boldsymbol{\alpha}^S, \mathbf{x}^S \rangle = & \sum\limits_{\mathbf{d} \in D(S)} \cvxCombPar_\mathbf{d} \langle \boldsymbol{\alpha}^S, \mathbf{p}^S(\mathbf{d}) \rangle \notag	\\
\geq & \min\limits_{\mathbf{d} \in D(S)} \langle \boldsymbol{\alpha}^S, \mathbf{p}^S(\mathbf{d}) \rangle, 
\end{align}
which is contradicted to (\ref{equ_alpha_S_q_S}). Therefore, $\mathbf{q}^S \notin \text{int}^S(B^S)$ and thus $\mathbf{q} \in \text{cl}(\concaveHull)$, implying that $R_{\text{IB}} \subseteq \text{cl}(\concaveHull)$. 

{\bf Part II of the proof: }

For the second part we show $\text{int}(\concaveHull) \subseteq R_{\text{IB}}$.

Given $\mathbf{q} \in \text{int}(R)$, by (\ref{align_int_R}), we know for all $S$, $\mathbf{q}^S \in \text{int}^S(C^S) \setminus B^S$. The goal is to show $\mathbf{q} \in R_{\text{IB}}$, i.e., to find an $\boldsymbol{\alpha} \succ \mathbf{0}$ such that for all subsets of users $S$, 
\begin{align}
\label{align_alpha_S_q_S}
\langle \boldsymbol{\alpha}^S, \mathbf{q}^S \rangle \leq \min\limits_{\mathbf{d} \in D(S)} \langle \boldsymbol{\alpha}^S, \mathbf{p}^S(\mathbf{d}) \rangle. 
\end{align}

We start with the following lemma which is proved in Appendix \ref{pf_lemma_finding_nonnegative_beta} in the sequel. 

\begin{lemma}
\label{lemma_finding_nonnegative_beta}
If \modelSatisfyPriorityImplication, given $\mathbf{q} \in \text{int}(\concaveHull)$, for each subset of users $S$, there exists nonzero $\boldsymbol{\beta}^S \succeq \mathbf{0}$, such that for all $S^\prime \subseteq S$ where $\boldsymbol{\beta}^{S^\prime} \neq \mathbf{0}$,  
\begin{align}
\label{align_beta_strict_smaller_than_condition}
\langle \boldsymbol{\beta}^{S^\prime}, \mathbf{q}^{S^\prime} \rangle < \min\limits_{\mathbf{d} \in D(S^\prime)} \langle \boldsymbol{\beta}^{S^\prime},  \mathbf{p}^{S^\prime}(\mathbf{d}) \rangle. 
\end{align}
\end{lemma}

Given Lemma \ref{lemma_finding_nonnegative_beta}, by letting $S = \fullUserSet$ we find $\boldsymbol{\beta}$ that is very similar to the $\boldsymbol{\alpha}$ we are looking for, except for two differences: (1) $\boldsymbol{\beta}$ may not be strictly positive, and (2) it is strictly ``less than'' in (\ref{align_beta_strict_smaller_than_condition}). 
The idea is to add a small perturbation to $\boldsymbol{\beta}$ to construct a strictly positive vector. 

Formally, given Lemma \ref{lemma_finding_nonnegative_beta}, to show $\reqvec \in R_\text{IB}$ we shall prove the following even stronger statement by induction. 

{\bf Claim 1:} If \modelSatisfyPriorityImplication, given $\mathbf{q} \in \text{int}(\concaveHull)$, for each subset of users $S$, there exists $\boldsymbol{\alpha}^S \succS{S} \mathbf{0}$ such that for all $S^\prime \subseteq S$, 
\begin{align}
\label{align_alpha_S_prime_q_S_prime}
\langle \boldsymbol{\alpha}^{S^\prime}, \mathbf{q}^{S^\prime} \rangle \leq 
\min\limits_{\mathbf{d} \in D(S^\prime)} \langle \boldsymbol{\alpha}^{S^\prime}, \mathbf{p}^{S^\prime}(\mathbf{d}) \rangle. 
\end{align}

By Claim 1, let $S = \fullUserSet$, we can find $\boldsymbol{\alpha} \succ \mathbf{0}$ satisfying (\ref{align_alpha_S_q_S}) which implies $\mathbf{q} \in R_{\text{IB}}$. Therefore, it suffices to prove Claim 1. We prove this by induction on the cardinality $|S|$ of user set $S$. 

If $|S| = 0$, clearly Claim 1 is correct. 

Suppose Claim 1 is correct for all $S$ with $|S| \leq k-1$ where $k \geq 1$. Given an $S$ with $|S| = k$, by Lemma \ref{lemma_finding_nonnegative_beta}, we can find nonzero $\boldsymbol{\beta}^S \succeq \mathbf{0}$ satisfying the conditions in Lemma \ref{lemma_finding_nonnegative_beta}. We separate the set $S$ into two sets $S_1$ and $S_2$ where
\begin{align*}
\beta_i^S > 0, i\in S_1,	\\
\beta_i^S = 0, i\in S_2. 
\end{align*}

Since $\boldsymbol{\beta}^S \neq \mathbf{0}$, $|S_2| \leq |S|-1 = k-1$. By induction of Claim 1 on $S_2$, there exists $\boldsymbol{\gamma}^{S_2} \succS{S_2} \mathbf{0}$ such that for any $S^\prime \subseteq S_2$, 
\begin{align}
\label{align_condition_for_S_2}
\langle \boldsymbol{\gamma}^{S^\prime}, \mathbf{q}^{S^\prime} \rangle \leq 
\min\limits_{\mathbf{d} \in D(S^\prime)} \langle \boldsymbol{\gamma}^{S^\prime}, \mathbf{p}^{S^\prime}(\mathbf{d}) \rangle. 
\end{align}

We claim that for small enough $\delta > 0$, 
$$
\boldsymbol{\alpha}^S = \boldsymbol{\beta}^S + \delta \boldsymbol{\gamma}^{S_2} \succS{S} \mathbf{0}
$$
satisfies condition (\ref{align_alpha_S_prime_q_S_prime}) for all $S^\prime \subseteq S$. 

Any $S^\prime \subseteq S$ falls into one of the following two cases: $S^\prime \subseteq S_2$ and $S^\prime \not\subseteq S_2$. It suffices to show (\ref{align_alpha_S_prime_q_S_prime}) in each case. 

If $S^\prime \subseteq S_2$, then $\boldsymbol{\alpha}^{S^\prime} = \boldsymbol{\gamma}^{S^\prime}$. By (\ref{align_condition_for_S_2}), we know (\ref{align_alpha_S_prime_q_S_prime}) is correct. 

If $S^\prime \not\subseteq S_2$, then $\boldsymbol{\beta}^{S^\prime} \neq \mathbf{0}$. Let $\boldsymbol{\gamma}^{S^\prime} = ({\boldsymbol{\gamma}^{S_2}})^{S^\prime}$. We know
$$
\langle \boldsymbol{\alpha}^{S^\prime}, \mathbf{q}^{S^\prime} \rangle = \langle \boldsymbol{\beta}^{S^\prime}, \mathbf{q}^{S^\prime} \rangle + \delta \langle \boldsymbol{\gamma}^{S^\prime}, \mathbf{q}^{S^\prime} \rangle. 
$$

By Lemma \ref{lemma_finding_nonnegative_beta}, $\langle \boldsymbol{\beta}^{S^\prime}, \mathbf{q}^{S^\prime} \rangle < \min\limits_{\mathbf{d} \in D(S^\prime)} \langle \boldsymbol{\beta}^{S^\prime},  \mathbf{p}^{S^\prime}(\mathbf{d}) \rangle$. Since there are finite subsets $S^\prime$, for small enough $\delta$, 

\begin{align*}
\langle \boldsymbol{\alpha}^{S^\prime}, \mathbf{q}^{S^\prime} \rangle \leq & \min\limits_{\mathbf{d} \in D(S^\prime)} \langle \boldsymbol{\beta}^{S^\prime},  \mathbf{p}^{S^\prime}(\mathbf{d}) \rangle	\\
\leq & \min\limits_{\mathbf{d} \in D(S^\prime)} \langle \boldsymbol{\alpha}^{S^\prime},  \mathbf{p}^{S^\prime}(\mathbf{d}) \rangle, 
\end{align*}
i.e., (\ref{align_alpha_S_prime_q_S_prime}) holds true. 

In summary, this proves Claim 1 and thus $\reqvec \in R_\text{IB}$. Therefore, $\text{int}(R) \subseteq R_\text{IB}$. 

\subsection{Proof of Lemma \ref{lemma_q_S_in_C_S}}
\label{pf_lemma_q_S_in_C_S}
First we introduce a further notation. Given a decision $\mathbf{d}$ and a user set $S$, we let $m(\mathbf{d}, S)$ represent the decision that satisfies
\begin{longitem}
\item $m(\mathbf{d}, S) \in D(S)$. 
\item For users $i, j \in S$ or $i, j \notin S$, if $i$ has higher priority than $j$ in $\mathbf{d}$, then $i$ also has higher priority than $j$ in decision $m(\mathbf{d}, S)$. 
\end{longitem}
In other words, $m(\mathbf{d}, S)$ is the priority decision obtained by modifying decision $\mathbf{d}$ to assign highest priorities to users in $S$ without changing the relative orders in and out of $S$, respectively. 

Given that the system satisfies monotonicity in payoffs, for all $i \in S$, we have that
\begin{align}
\label{align_m_d_S}
p_i(m(\mathbf{d}, S)) \geq p_i(\mathbf{d}). 
\end{align}

Given $\reqvec \in C$, the goal is to show $\reqvec^S \in C^S$ for all subsets of users $S$. 
By definition of $C$, there exists a convex combination of vectors in $P$ that dominates $\reqvec$, i.e., there exists $\{\alpha_\mathbf{d} | \mathbf{d} \in D\}$ such that 
$$
\reqvec \preceq \sum\limits_{\mathbf{d} \in D} \alpha_\mathbf{d} \mathbf{p}(\mathbf{d}), 
$$
and $\alpha_\mathbf{d} \geq 0, \sum\limits_{\mathbf{d} \in D} \alpha_\mathbf{d} = 1$. 

Therefore, for any subset of users $S$, 
$$
\reqvec^S \preceq \sum\limits_{\mathbf{d} \in D} \alpha_\mathbf{d} \mathbf{p}^S(\mathbf{d}), 
$$
which by (\ref{align_m_d_S}) gives
$$
\reqvec^S \preceq \sum\limits_{\mathbf{d} \in D} \alpha_\mathbf{d} \mathbf{p}^S(m(\mathbf{d}, S)). 
$$

We let $\mathbf{x}^S = \sum\limits_{\mathbf{d} \in D} \alpha_\mathbf{d} \mathbf{p}^S(m(\mathbf{d}, S))$. Since $m(\mathbf{d}, S) \in D(S)$, we know $\mathbf{x}^S \in \text{Conv}(P^S)$ , and by definition of $C^S$, 
$$
\reqvec^S \in C^S. 
$$

\subsection{Proof of Lemma \ref{lemma_finding_nonnegative_beta}}
\label{pf_lemma_finding_nonnegative_beta}
Given $\mathbf{q} \in \text{int}(R)$, by (\ref{align_int_R}) we know that for all subsets of users $S$, $\mathbf{q}^S \in \text{int}^S(C^S) \setminus B^S$. 
Given a subset of users $S$, the goal is to find $\boldsymbol{\beta}^S$ which satisfies the requirements in Lemma \ref{lemma_finding_nonnegative_beta}. 
In this proof, we focus on the subspace of $S$. 

Since $\mathbf{q}^S \in \text{int}^S(C^S)$, there exists $\mathbf{x}^S \in \text{Conv}(P^S)$ such that $\mathbf{q}^S \precS{S} \mathbf{x}^S$. Since $\mathbf{q}^S \notin B^S$ and $\mathbf{x}^S \in \text{Conv}(P^S) \subseteq B^S$, we know that connecting $\mathbf{q}^S$ and $\mathbf{x}^S$ intersects $\text{bd}^S(B^S)$ at some point denoted by $\mathbf{v}^S$. By the closure property of $B^S$, $\mathbf{v}^S \in B^S$ and thus, $\mathbf{v}^S \succS{S} \mathbf{q}^S$. 
Since $\mathbf{v}^S \in \text{bd}^S(B^S)$, we get that $\mathbf{v}^S$ lies on a supporting hyperplane \cite{BoV09} of $B^S$, and by definition of $B^S$, there exists nonzero normal vector $\boldsymbol{\beta}^S \succeq \mathbf{0}$ of this supporting hyperplane such that
\begin{align}
\label{align_v_is_minimum}
\langle \boldsymbol{\beta}^S, \mathbf{v}^S \rangle = \min\limits_{\mathbf{d} \in D(S)} \langle \boldsymbol{\beta}^S, \mathbf{p}^S(\mathbf{d}) \rangle. 
\end{align}

Figure~\ref{fig_finding_beta_in_pf_Lemma_4} conceptually shows the process of constructing $\boldsymbol{\beta}^S$. The circles represent the expected payoff vectors. For simplicity we suppress the superscript $S$ in the figure. 
We shall show this $\boldsymbol{\beta}^S$ satisfies the requirements in Lemma \ref{lemma_finding_nonnegative_beta}. 

\begin{figure}[htp]
  \centering
  \includegraphics[width=0.65\textwidth]{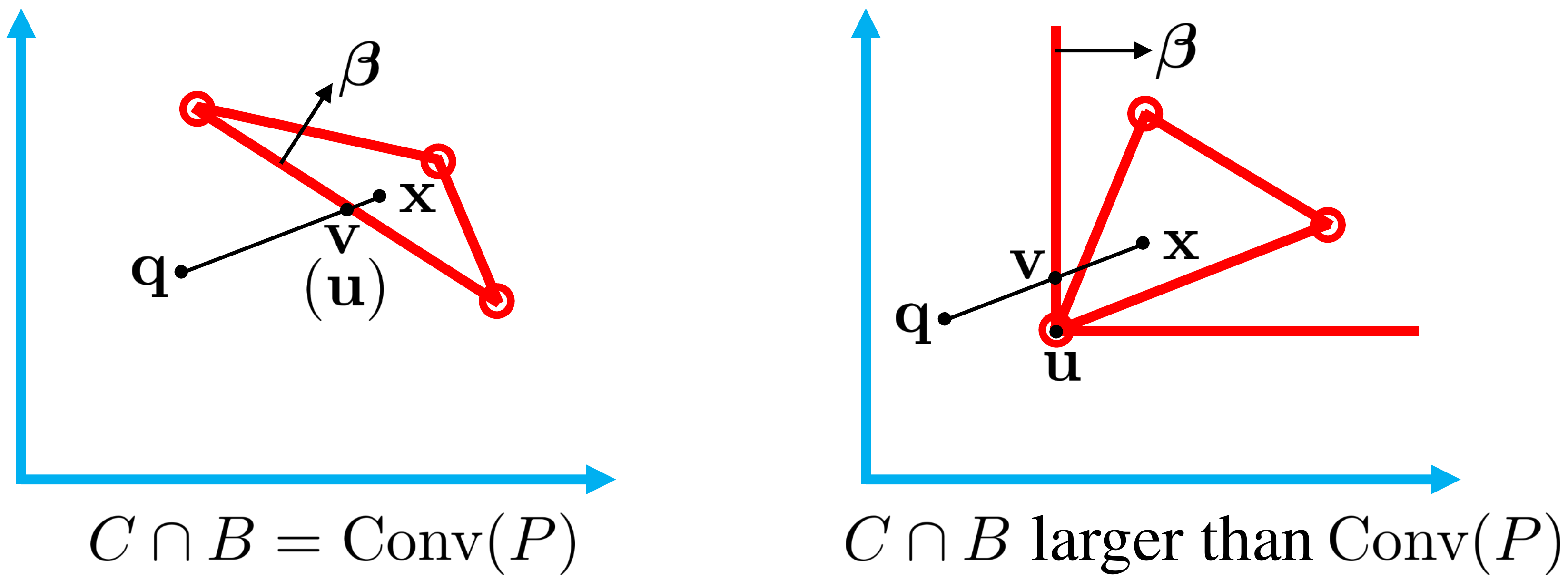}
  \caption{The process of constructing $\boldsymbol{\beta}$ when $C\cap B$ equals to $\text{Conv}(P)$ (left figure), and when $C \cap B$ is larger than $\text{Conv}(P)$ (right figure). }
  \label{fig_finding_beta_in_pf_Lemma_4}
\end{figure}

Since $\mathbf{v}^S \in \text{bd}^S(B^S) \subseteq B^S$, there exists $\mathbf{u}^S \in \text{Conv}(P^S)$ such that $\mathbf{v}^S \succeq \mathbf{u}^S$, and by (\ref{align_v_is_minimum}), we get that
$$
\langle \boldsymbol{\beta}^S, \mathbf{u}^S \rangle \leq \min\limits_{\mathbf{d} \in D(S)} \langle \boldsymbol{\beta}^S, \mathbf{p}^S(\mathbf{d}) \rangle. 
$$
The vector $\mathbf{u}^S$ is also shown in Figure{~\ref{fig_finding_beta_in_pf_Lemma_4}}. 

On the other hand, since $\mathbf{u}^S \in \text{Conv}(P^S)$, by similar analysis as in (\ref{align_cvx_greater_than_min}) we know that
$$
\langle \boldsymbol{\beta}^S, \mathbf{u}^S \rangle \geq \min\limits_{\mathbf{d} \in D(S)} \langle \boldsymbol{\beta}^S, \mathbf{p}^S(\mathbf{d}) \rangle. 
$$

Thus, 
$$
\langle \boldsymbol{\beta}^S, \mathbf{u}^S \rangle = \langle \boldsymbol{\beta}^S, \mathbf{v}^S \rangle = \min\limits_{\mathbf{d} \in D(S)} \langle \boldsymbol{\beta}^S, \mathbf{p}^S(\mathbf{d}) \rangle, 
$$
which implies $u_i^S = v_i^S$ if $\beta_i^S \neq 0$. Since $\mathbf{v}^S \succS{S} \mathbf{q}^S$, for all subset $S^\prime \subseteq S$ where $\boldsymbol{\beta}^{S^\prime} \neq \mathbf{0}$, we have that
$$
\langle \boldsymbol{\beta}^{S^\prime}, \mathbf{q}^{S^\prime} \rangle < \langle \boldsymbol{\beta}^{S^\prime}, \mathbf{u}^{S^\prime} \rangle, 
$$

Therefore, to show (\ref{align_beta_strict_smaller_than_condition}) it suffices to show for all $S^\prime \subseteq S$ where $\boldsymbol{\beta}^{S^\prime} \neq \mathbf{0}$ that
\begin{align}
\label{align_beta_S_prime_u_S_prime}
\langle \boldsymbol{\beta}^{S^\prime}, \mathbf{u}^{S^\prime} \rangle \leq \min\limits_{\mathbf{d} \in D(S^\prime)} \langle \boldsymbol{\beta}^{S^\prime},  \mathbf{p}^{S^\prime}(\mathbf{d}) \rangle. 
\end{align}

Given $\mathbf{u}^S \in \text{Conv}(P^S)$, we write $\mathbf{u}^S = \sum\limits_{\mathbf{d} \in D(S)}\cvxCombPar_\mathbf{d} \mathbf{p}^S(\mathbf{d})$ where $\cvxCombPar_\mathbf{d} \geq 0$ and $\sum\limits_{\mathbf{d} \in D(S)} \cvxCombPar_\mathbf{d} = 1$. 

For all $S^\prime \subseteq S$ where $\boldsymbol{\beta}^{S^\prime} \neq \mathbf{0}$, we can rewrite $\langle \boldsymbol{\beta}^{S^\prime}, \mathbf{u}^{S^\prime} \rangle$ as follows,
\begin{eqnarray*}
\langle \boldsymbol{\beta}^{S^\prime}, \mathbf{u}^{S^\prime} \rangle & = & \langle \boldsymbol{\beta}^{S}, \mathbf{u}^{S} \rangle - \langle \boldsymbol{\beta}^{S \setminus S^\prime}, \mathbf{u}^{S \setminus S^\prime} \rangle	\\
& = & \min\limits_{\mathbf{d} \in D(S)} \langle \boldsymbol{\beta}^S, \mathbf{p}^S(\mathbf{d}) \rangle - \sum\limits_{\mathbf{d} \in D(S)} \cvxCombPar_\mathbf{d} \langle \boldsymbol{\beta}^{S \setminus S^\prime}, \mathbf{p}^{S \setminus S^\prime}(\mathbf{d}) \rangle	\\
& \leq & \min\limits_{\mathbf{d} \in D(S)} \langle \boldsymbol{\beta}^S, \mathbf{p}^S(\mathbf{d}) \rangle - \min\limits_{\mathbf{d} \in D(S)} \langle \boldsymbol{\beta}^{S \setminus S^\prime}, \mathbf{p}^{S \setminus S^\prime}(\mathbf{d}) \rangle. 
\end{eqnarray*}

To show (\ref{align_beta_S_prime_u_S_prime}), it suffices to show that
\begin{eqnarray}
\lefteqn{\min\limits_{\mathbf{d} \in D(S)} \langle \boldsymbol{\beta}^S, \mathbf{p}^S(\mathbf{d}) \rangle} \notag \\
& \leq & \min\limits_{\mathbf{d} \in D(S)} \langle \boldsymbol{\beta}^{S \setminus S^\prime}, \mathbf{p}^{S \setminus S^\prime}(\mathbf{d}) \rangle \label{align_sum_S_minus_S_prime}	\\
& & + \min\limits_{\mathbf{d} \in D(S^\prime)} \langle \boldsymbol{\beta}^{S^\prime},  \mathbf{p}^{S^\prime}(\mathbf{d}) \rangle. \label{align_sum_S_prime}
\end{eqnarray}

Suppose $\mathbf{d}_1 \in D(S)$ and $\mathbf{d}_2 \in D(S^\prime)$ are the optimal solutions for (\ref{align_sum_S_minus_S_prime}) and (\ref{align_sum_S_prime}), respectively. Formally, 
$$
\langle \boldsymbol{\beta}^{S \setminus S^\prime}, \mathbf{p}^{S \setminus S^\prime}(\mathbf{d}_1) \rangle = \min\limits_{\mathbf{d} \in D(S)} \langle \boldsymbol{\beta}^{S \setminus S^\prime}, \mathbf{p}^{S \setminus S^\prime}(\mathbf{d}) \rangle, 
$$
$$
\langle \boldsymbol{\beta}^{S^\prime},  \mathbf{p}^{S^\prime}(\mathbf{d}_2) \rangle = \min\limits_{\mathbf{d} \in D(S^\prime)} \langle \boldsymbol{\beta}^{S^\prime},  \mathbf{p}^{S^\prime}(\mathbf{d}) \rangle. 
$$

We consider the unique decision $\mathbf{d}_3$ that satisfies the following: 
First, $\mathbf{d}_3\in D(S^\prime)$, i.e., $\mathbf{d}_3$ assigns highest priority to users in $S^\prime$. 
Second, the priority ordering for user subset $S^\prime$ in $\mathbf{d}_3$ are the same as those in $\mathbf{d}_2$. 
Third, the priority ordering for user subset $\fullUserSet \setminus S^\prime$ in $\mathbf{d}_3$ are the same as those in $\mathbf{d}_1$. 

Since $\mathbf{d}_1 \in D(S)$, we know $\mathbf{d}_3 \in D(S)$ and therefore, 
$$
\min\limits_{\mathbf{d} \in D(S)} \langle \boldsymbol{\beta}^S, \mathbf{p}^S(\mathbf{d}) \rangle \leq \langle \boldsymbol{\beta}^S, \mathbf{p}^S(\mathbf{d}_3) \rangle. 
$$

Now it suffices to show that
\begin{align*}
\langle \boldsymbol{\beta}^S, \mathbf{p}^S(\mathbf{d}_3) \rangle \leq \langle \boldsymbol{\beta}^{S \setminus S^\prime}, \mathbf{p}^{S \setminus S^\prime}(\mathbf{d}_1) \rangle + \langle \boldsymbol{\beta}^{S^\prime},  \mathbf{p}^{S^\prime}(\mathbf{d}_2) \rangle. 
\end{align*}

This is true because given that the system satisfies monotonicity in payoffs, we can get that
$$
p_i^{S^\prime}(\mathbf{d}_3) \leq p_i^{S^\prime}(\mathbf{d}_2) \text{~for~} i \in S^\prime,  
$$
and
$$
p_i^{S\setminus S^\prime}(\mathbf{d}_3) \leq p_i^{S\setminus S^\prime}(\mathbf{d}_1) \text{~for~} i \in S\setminus S^\prime. 
$$

In summary, this proves (\ref{align_beta_S_prime_u_S_prime}) and thus $\boldsymbol{\beta}^S$ satisfies the conditions in Lemma \ref{lemma_finding_nonnegative_beta}. 

\subsection{Proof of Theorem \ref{thm_sufficient_condition_for_LDF_optimality}}
\label{appendix_pf_thm_sufficient_condition_for_LDF_optimality}

Clearly $\feasibilityRegion_{\mathbf{w}\text{-LDF}} \subseteq \feasibilityRegion \subseteq \text{cl}(C)$. 
To show $\text{int}(C) \subseteq \feasibilityRegion_{\mathbf{w}\text{-LDF}}$, by (\ref{align_A_vs_R_LDF}) it suffices to show $\text{int}(C) \subseteq \text{int}(\concaveHull)$. 

Given $\reqvec \in \text{int}(C)$, the goal is to show $\reqvec \in \text{int}(\concaveHull)$, i.e., for all user subsets $S$, 
$$
\reqvec^S \in \text{int}^S(C^S)\setminus B^S. 
$$

Given a user subset $S$, by Lemma \ref{lemma_q_S_in_C_S} we know $\reqvec^S \in C^S$. Further we can show $\reqvec^S \in \text{int}^S(C^S)$ since otherwise $\reqvec \in \text{bd}(C)$. By definition of interior and $C^S$, there exists $\mathbf{x}^S \in \text{Conv}(P^S)$ such that $\reqvec^S \precS{S} \mathbf{x}^S$.  

Suppose $\reqvec^S \in B^S$, by definition there exists $\mathbf{y}^S \in \text{Conv}(P^S)$ such that $\reqvec^S \succeq \mathbf{y}^S$. Now we get two vectors $\mathbf{x}^S, \mathbf{y}^S \in \text{Conv}(P^S)$ and $\mathbf{x}^S \succS{S} \mathbf{y}^S$. Since vectors in $P^S$ lie on a hyperplane and $\mathbf{x}^S, \mathbf{y}^S \in \text{Conv}(P^S)$, there exists nonzero $\boldsymbol{\alpha}^S \succeq \mathbf{0}$ such that
$$
\langle \boldsymbol{\alpha}^S, \mathbf{x}^S \rangle = \langle \boldsymbol{\alpha}^S, \mathbf{y}^S \rangle, 
$$
which contradicts with $\mathbf{x}^S \succS{S} \mathbf{y}^S$. 

Therefore, $\reqvec^S \notin B^S$ and thus $\text{int}(C) \subseteq \text{int}(\concaveHull)$. 

\subsection{Proof of Theorem \ref{thm_efficiency_ratio}}
\label{appendix_pf_thm_efficiency_ratio}

Given monotonicity in payoffs, by (\ref{align_A_vs_R_LDF}) we know $\text{int}(R) \subseteq \feasibilityRegion_{\mathbf{w}\text{-LDF}}$ and $\feasibilityRegion$ differs from $C$ by at most a boundary. By the definition of the efficiency ratio, we know
\begin{align*}
\gamma_{\mathbf{w}\text{-LDF}} & = \sup\{\gamma | \gamma \feasibilityRegion \subseteq \feasibilityRegion_{\mathbf{w}\text{-LDF}} \} \geq \sup\{\gamma | \gamma C \subseteq \text{int}(R) \} = \sup\{\gamma | \gamma C \subseteq \text{cl}(R) \}. 
\end{align*}

To show $\gamma_{\mathbf{w}\text{-LDF}} \geq \min\limits_{S \subseteq \fullUserSet}\sigma_S$, it suffices to show $\sup\{\gamma | \gamma C \subseteq \text{cl}(R) \} \geq \min\limits_{S \subseteq \fullUserSet}\sigma_S$, i.e., for each $\reqvec \in C$, we have $\min\limits_{S \subseteq \fullUserSet}\sigma_S \cdot \reqvec \in \text{cl}(R)$, which is equivalent to showing that for each $\reqvec \in C$, there exists a subset of users $S \subseteq \fullUserSet$, such that $\sigma_S \cdot \reqvec \in \text{cl}(R)$. 

Given a $\reqvec \in C$, we define $\lambda(\reqvec, R) = \sup\{ \lambda | \lambda \reqvec \in R \}$ which represents how far the vector $\reqvec$ can extend before it goes beyond the region $R$ and let $\reqvec_{(R)} = \lambda(\reqvec, R) \cdot \reqvec$. We claim there exists a user subset $S$ such that $\reqvec_{(R)}^{S} \in \text{bd}^{S}(B^{S})$ since otherwise we can increase $\lambda(\reqvec, R)$ while guaranteeing $\reqvec_{(R)}$ is still in $R$. 
Next we shall show $\sigma_S \cdot \reqvec \in \text{cl}(R)$, i.e., $\lambda(\reqvec, R) \geq \sigma_S$. 

Since $\reqvec_{(R)}^{S} \in \text{bd}^{S}(B^{S}) \subseteq B^{S}$, there exists $\mathbf{x}^{S} \in \text{Conv}(P^{S})$ such that $\reqvec_{(R)}^{S} \succeq \mathbf{x}^{S}$. Since $\reqvec \in C$, by Lemma \ref{lemma_q_S_in_C_S} we know $\reqvec^{S} \in C^{S}$ and thus, there exists $\mathbf{y}^{S} \in \text{Conv}(P^{S})$ such that $\reqvec^{S} \preceq \mathbf{y}^{S}$. 

Given that $\reqvec_{(R)} = \lambda(\reqvec, R) \cdot \reqvec$, for any nonzero $\boldsymbol{\alpha}^{S} \succeq \mathbf{0}$, we have that
\begin{align*}
\langle \reqvec_{(R)}^{S}, \boldsymbol{\alpha}^{S}  \rangle = \lambda(\reqvec, R) \cdot \langle \reqvec^{S}, \boldsymbol{\alpha}^{S}  \rangle. 
\end{align*}

By $\reqvec_{(R)}^{S} \succeq \mathbf{x}^{S}$ and $\reqvec^{S} \preceq \mathbf{y}^{S}$, we get that
\begin{align*}
\langle \mathbf{x}^{S}, \boldsymbol{\alpha}^{S}  \rangle \leq \lambda(\reqvec, R) \cdot \langle \mathbf{y}^{S}, \boldsymbol{\alpha}^{S} \rangle. 
\end{align*}

Since $\mathbf{x}^{S} \in \text{Conv}(P^{S})$, by similar analysis as in (\ref{align_cvx_greater_than_min}) we know $\langle \mathbf{x}^{S}, \boldsymbol{\alpha}^{S}  \rangle \geq \min\limits_{\mathbf{d} \in D(S)} \langle \boldsymbol{\alpha}^{S}, \mathbf{p}^{S}(\mathbf{d}) \rangle $. Similarly we can show $\langle \mathbf{y}^{S}, \boldsymbol{\alpha}^{S}  \rangle \leq \max\limits_{\mathbf{d} \in D(S)} \langle \boldsymbol{\alpha}^{S}, \mathbf{p}^{S}(\mathbf{d}) \rangle$. Thus, 
\begin{align}
\label{align_min_projection_leq_ratio_max_projection}
\min\limits_{\mathbf{d} \in D(S)} \langle \boldsymbol{\alpha}^{S}, \mathbf{p}^{S}(\mathbf{d}) \rangle \leq \lambda(\reqvec, R) \cdot \max\limits_{\mathbf{d} \in D(S)} \langle \boldsymbol{\alpha}^{S}, \mathbf{p}^{S}(\mathbf{d}) \rangle. 
\end{align}
Clearly, $\max\limits_{\mathbf{d} \in D(S)} \langle \boldsymbol{\alpha}^{S}, \mathbf{p}^{S}(\mathbf{d}) \rangle > 0$ for any nonzero $\boldsymbol{\alpha}^{S} \succeq \mathbf{0}$. 

Since (\ref{align_min_projection_leq_ratio_max_projection}) is true for any nonzero $\boldsymbol{\alpha}^{S} \succeq \mathbf{0}$, we get that
$$
\lambda(\reqvec, R) \geq 
\max\limits_{{\boldsymbol{\alpha}^S \succeq \mathbf{0}} \atop {\boldsymbol{\alpha}^S \neq \mathbf{0}}}
\frac{\min\limits_{\mathbf{d} \in D(S)} \langle \boldsymbol{\alpha}^{S}, \mathbf{p}^{S}(\mathbf{d}) \rangle}
{\max\limits_{\mathbf{d} \in D(S)} \langle \boldsymbol{\alpha}^{S}, \mathbf{p}^{S}(\mathbf{d}) \rangle}
= \sigma_{S}. 
$$

Therefore, for each $\reqvec \in C$, there exists a subset of users $S \subseteq \fullUserSet$ such that $\sigma_S \cdot \reqvec \in \text{cl}(R)$, and thus,
$\gamma_{\mathbf{w}\text{-LDF}} \geq \min\limits_{S \subseteq \fullUserSet} \sigma_S$. 

\subsection{Proof of Corollary \ref{corollary_exchangeable expected_workloads}}
\label{appendix_pf_corollary_exchangeable expected_workloads}
By Theorem \ref{thm_sufficient_condition_for_LDF_optimality}, to show $\mathbf{w}$-LDF policies are feasibility optimal, it suffices to show the system satisfies subset payoff equivalence. To show this, it suffices to show for all user subsets $S \subseteq \fullUserSet$ and all priority decisions $\mathbf{d}_1, \mathbf{d}_2 \in D(S)$ that 
$$
\langle \mathbf{1}^S, \mathbf{p}^S(\mathbf{d}_1) \rangle = \langle \mathbf{1}^S, \mathbf{p}^S(\mathbf{d}_2) \rangle. 
$$

This is true because we can convert $\mathbf{d}_1$ to $\mathbf{d}_2$ by repeatedly switching a pair of users in $\mathbf{d}_1$ at each step such that at step $k$ both decisions assign the highest $k$ priorities to the same users, respectively. By the definition of exchangeable expected payoffs, the sum of the expected payoffs for users in $S$ remains the same at each step.

\dissertationStart
\appendix
\section{Some Further Explanations}
\subsection{TODO for this paper}
\begin{enumerate}
\item Make the shadowed areas in the figures thicker. 
\item Can I use package multirow? Change looking for tables. 
\item Work on defn of feasibility, i.r., p.r.. 
\item Finish proof of corollary 1. Refer to 6.28 - 7.3 page 43. 
\item We should explain that LDF is not optimal. Maybe give the example that we use in the group presentation. 
\item stand on the reader's point of view: is it clear? 
\item What we discuss in the beginning has to be super clear and convincing. 
\item If we discuss something in the sequel, mention that! E.g., the simulation setup if it is incomplete. 
\item Re-go over comments for sigmetrics review. 
\item Are we missing any simple generalization? E.g., energy consumption. 
\item Improve system model using ``lattice''. 
\item generalization of payoff model: considering iid and Markovian arrivals, and the scheduler knows about the arrival for current period. 
\item consider to generalize general payoff model to $\mathbf{p}(\mathbf{X})$ not just $\mathbf{p}(\mathbf{d})$. 
\item Re-go over comments for sigmetrics review. 
\end{enumerate}

\subsection{TODO for dissertation}
\begin{enumerate}
\item Give examples to illustrate the difference between our sufficient condition and Walrand's work. 
\end{enumerate}
\commentEnd\fi

\end{document}